\documentclass[12pt, onecolumn]{IEEEtran}

\usepackage[T1]{fontenc}

\usepackage[british]{babel}
\usepackage{hyphenat}
\usepackage{booktabs} 
\usepackage{graphicx,multicol,latexsym,amsmath,amssymb}
\usepackage{color}
\usepackage[dvipsnames]{xcolor}
\usepackage{amsthm}
\usepackage{subcaption}
\newcommand{\beq}{\begin{equation}}
\newcommand{\eeq}{\end{equation}}

\newcommand{\commD}[1]{\marginpar{%
\begin{color}{red}
\vskip-\baselineskip 
\raggedright\footnotesize
\itshape\hrule \smallskip D: #1\par\smallskip\hrule\end{color}}}

\newcommand{\setT}{\mathcal{T}}
\newcommand{\E}{\mathbbm{E}}

\newcommand{\expect}[1]{{\mathbb E}\left\{{\displaystyle #1}\right\}}

\def\Pm{P_{\max}}
\newtheorem{theorem}{Theorem}[section]

\newtheorem{corollary}{Corollary}
\newtheorem{definition}{Definition}
\newtheorem{lemma}{Lemma}

\newtheorem{remark}{Remark}

\newif\ifnotesw
\noteswtrue

\newif\ifnotesw
\noteswtrue

\def\CL{C_l}
\def\CS{{\cal S}}

\def\E{\mathbf{E}}

\def\CI{{\cal I}}

\usepackage{amsmath}
\usepackage[cmintegrals]{newtxmath}
\begin{document}
%
\title{Characterizing the Energy Trade-Offs of End-to-End Vehicular Communications using an Hyperfractal Urban Modelling }
%
%
%

\author{Dalia~Popescu, 
        Philippe~Jacquet,~\IEEEmembership{Fellow,~IEEE,}
        Bernard~Mans, 
        Bart{\l}omiej~B{\l}aszczyszyn
        
\thanks{B.  Mans  was  supported  in  part  by  the  Australian  Research  Council  under Grant DP170102794.

Part of the work has been done at Lincs.
Dalia Popescu was with Nokia Bell Labs,  91620 Nozay, France. Bernard Mans is with Macquarie University, Sydney, Australia (bernard.mans@mq.edu.au). 
Philippe Jacquet and Bart{\l}omiej~B{\l}aszczyszyn are with INRIA, France.}
}

%
%


\maketitle

\begin{abstract}

We characterize trade-offs between the end-to-end communication delay and the energy in urban vehicular communications with infrastructure assistance. 
Our study exploits the self-similarity of the location of communication entities in cities by modeling them  with an innovative model called ``hyperfractal''. We show that 
the hyperfractal model can be extended to incorporate road-side infrastructure and provide stochastic geometry tools to allow a rigorous analysis. 
We compute theoretical bounds for the end-to-end communication hop count considering two different energy-minimizing goals: either total accumulated energy or maximum energy per node. 
We prove that the hop count for an end-to-end transmission is bounded by $O(n^{1-\alpha/(d_F-1)})$ where $\alpha<1$ and $d_F>2$ is the fractal dimension of the mobile nodes process.
This proves that for both constraints the energy decreases as we allow choosing routing paths of higher length. The asymptotic limit of the energy becomes significantly small when
 the number of nodes becomes asymptotically large.
A lower bound on the network throughput capacity with constraints on path energy is also  given. We show that our model fits real deployments where open data sets are available.
 The results are confirmed through 
 simulations using different fractal dimensions in a Matlab simulator.
\end{abstract}
\maketitle


%
\IEEEpeerreviewmaketitle

\section{Introduction}
\subsection{Motivation and Background}
Vehicular communications, V2V (vehicle-to-vehicle), V2I (vehicle to infrastructure) or V2X (vehicle to everything), are a key component of the 5th Generation (5G) and beyond communications. These `verticals' represent one major focus of the telecommunication industry nowadays. Yet like many innovation opportunities on the horizon they arrive with significant challenges. 
As the vehicular networks continue to scale up to reach tremendous network sizes with diverse hierarchical structures and node types, and as vehicular interactions become more complex with entities having hybrid functions and levels of intelligence and control, it is paramount to provide an effective integration of vehicular networks within the complex urban environment. 
Automated and autonomous driving in such a complex and evolving environment requires sensors that generate a huge amount of data demanding high bandwidth and data rates \cite{mm_vehicular2}. Furthermore, an effective integration of these new types of communication in the new radio ``babbling" created by the other 5G actors such as evolved mobile broadband, ultra-reliable low latency communications and massive machine-type communications, requires a careful design for optimal connectivity, low interference, and maximum security. 





5G NR (5th Generation New Radio) is essentially a multi-beam system, with high-frequency ranges generated by millimeter-wave (mmWave) technology  \cite{FR2}. With many GHz of spectrum to offer, millimeter-wave bands are the key for attaining the high capacity and services diversity of the NR. For a long time these frequencies have been disregarded for cellular communications due to their large near-field loss, and poor penetration (blocking) through
common material, yet recent research and experiments have shown that communications are feasible in ranges of 150-200 meters dense
urban scenarios with the use of such high gain directional
\cite{mmWave_3}. Furthermore, the tight requirements (e.g, line of sight, short-range) are easily answered as the embedding space of vehicular networks leads to a highly directive topology (as much as it is possible, roads are built as straight lines) \cite{Bartek}.

Given the numerous challenges of mmWave \cite{new_mm} and the important place the vehicular communications hold in the new communications era, realistic modeling of the topology for accurate estimation of network metrics is mandatory.  
The research community has proposed stochastic models that usually fit with high precision cellular networks or ad hoc networks. Yet for vehicles, and more importantly, for vehicles using mmWave technology, this cannot be done without taking into account the crucial fact that the effectiveness of the communications are influenced by the environmental topology. Cars are located on streets and streets are conditioned by a world-wide common urban architecture that has interesting features. 
One major feature of the urban architecture that we exploit in this work is self-similarity. 

While it has been extensively studied in diverse research fields such as biology and chemistry, self-similarity has been only recently introduced to wireless communications, after understanding that the device-to-device communication topologies follow the social human topology.
Self-similarity is present in every aspect of the surrounding environment but is particularly emphasized in the urban environment. The hierarchic organization with different degrees of scaling of cities is a perfect illustration of the fractal structure of human society \cite{Batty2008TheSS}. Figure \ref{fig:min_data} presents a snapshot of the traffic in a neighborhood of Minneapolis. Common patterns and hierarchical organizations can easily be identified in the traffic measurements and shall be further explained in this paper. 

\begin{figure}[httb]\centering
\includegraphics[scale=0.45, trim=0cm 0.5cm 0cm 0cm]{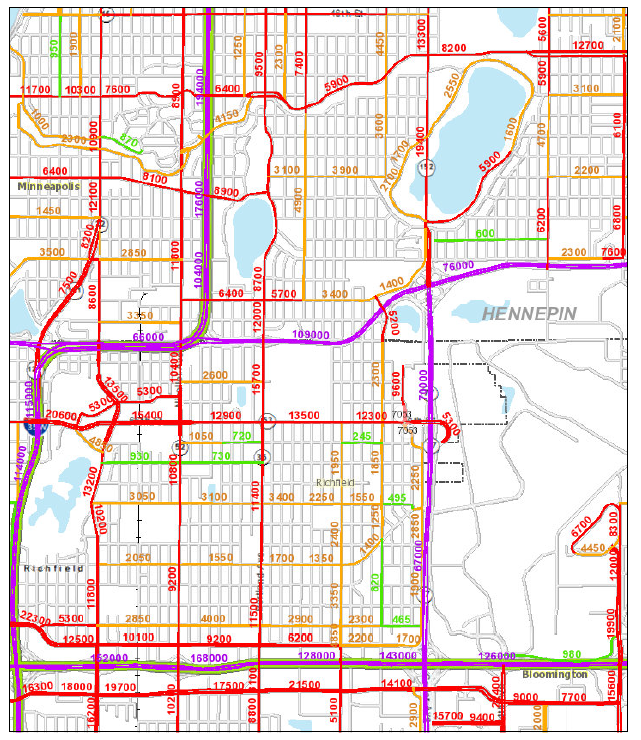}
\caption{Minneapolis traffic snapshot}
\vspace{-0.09cm}
\label{fig:min_data}
\end{figure}

In this paper, we extend the "hyperfractal" model that we have introduced in~\cite{spaswin,gsi} to better capture the impact of the network topology on the fundamental performance limits of end-to-end communications over vehicular networks in urban settings. The model consists of assigning decaying traffic densities to city streets, thus avoiding the extremes of  regularity (e.g. Manhattan grid) and uniform randomness (e.g. Poisson point process), the fitting of the model with traffic data of real cities having been showcased in~\cite{archiv17}. The hyperfractal model exploits the self-similarity: e.g., it is characterized by a dimension that is larger than the dimension of the euclidean dimension of the embedding space, that is larger than 2 when the whole network lays in a 2-dimensional plane.

Our previous results in~\cite{spaswin} revealed that, for nodes, the number of hops in a routing path between an arbitrary source-destination pair increases as a power function of the population $n$ of nodes when $n$ tends to infinity. However, we showed that the exponent tends to zero when the fractal dimension tends to infinity. 
An initial observation for this model is that the optimal path may have to go through streets of low density where inter-vehicle distance can become large, therefore the transmission becomes expensive in terms of energy cost. Hence, in this paper, the focus will be on the study of the relationship between efficient communications and energy costs. 

\subsection{Contributions and paper organization:} 
Our goal is to characterize trade-offs between the end-to-end communication delay and the energy in urban vehicular communications with infrastructure assistance in modern cities. 

Our first contribution is to exploit the self-similarity of the location of the traffic and vehicles in cities by modeling the communication entities and relationships with an innovative model called ``hyperfractal'' (to avoid extremes of classical Poisson distribution or uniform distribution tools) and to show that the hyperfractal model can be extended to incorporate road-side infrastructure, as relays. This provides fundamental properties and tools in the framework of stochastic geometry that allow for a rigorous analysis (and are of independent interest for other studies).

Our main contributions are theoretical bounds for the end-to-end communication hop count. We will consider two different energy-minimizing goals: (i) total accumulated energy or (ii) maximum energy per node. 
We will prove that the hop count for an end-to-end transmission is bounded by $O(n^{1-\alpha/(d_F-1)})$ where $\alpha<1$ and $d_F>2$ is the fractal dimension of the mobile nodes process, thus proving that for both constraints the energy decreases as we allow choosing routing paths of higher length. We will also show that the asymptotic limit of the energy becomes significantly small when
 the number of nodes becomes asymptotically large. This is also completed with a lower bound on the network throughput capacity with constraints on path energy. 
 Finally we will show that our model fits real deployments where open data sets are available. The results are confirmed through 
 simulations using different fractal dimensions and path loss coefficients, using a discrete-event simulator in Matlab.

The paper is organized as follows:
\begin{itemize}
\item In Section \ref{model}, we first enhance the hyperfractal model by taking into account mmWave communication range variations as well as the energy costs of transmission.  In addition, we enrich the model by incorporating road-side infrastructure with communication relays (with radio communication range variations). We exploit the self-similarity of intersection locations in urban settings.
\item  in Section \ref{properties}, theoretical properties of the hyperfractal model are obtained to allow the characterization of bounds within the communication model. These properties are developed within a classic stochastic geometric framework and are of interest on their own.
\item In Section \ref{results}, we prove that for an end-to-end transmission in a hyperfractal setup, the energy (either accumulated along the path or bounded for each node) decreases if we allow the path length to increase. In fact, we show that the asymptotic limit of the energy tends to zero when
$n$, the number of nodes, tends to infinity. 
We also prove a lower bound on the network throughput capacity with constraints on path energy.
\item  In Section \ref{calcul_dr}, we further provide a fitting procedure that allows computing the fractal dimension of the relay process using traffic lights data sets.
\item Finally, Section \ref{simulations} validates our analytical results using a discrete-time event-based simulator developed in Matlab.
\end{itemize}

\section{Related Works}





Millimeter-wave is a key technological brick of the 5G NR networks, as foreseen in the ground-breaking work done in \cite{mm5} and already proved by ongoing deployments. The research community has been already investigating challenges that may appear and proposing innovative solutions.
Vehicular communications are one of the areas that are to benefit from the high capacity offered by the mmWave technology.
 In \cite{new_veh_mm1}, the authors propose an information-centric network (ICN)-based mmWave vehicular framework together with a decentralized vehicle association algorithm to realize low-latency content disseminations. The study shows that the proposed algorithm can improve the content dissemination efficiency yet there are no consideration about the energy.  The purpose of \cite{cite_5} is optimizing energy efficiency in a cellular system with relays with D2D (device-to-device) communications using mmWave.

As mmWave is highly directional and blockages raise concerns, the authors of \cite{new_veh_mm3} propose an online learning algorithm addressing the problem of beam selection with environment-awareness in mmWave vehicular systems. The sensitivity to blockages is generally solved with the assistance of the relaying infrastructure. 
The authors of \cite{new_veh_mm2} attempt to solve the dependency of infrastructure for relaying in vehicular communications by exploiting social interactions. In \cite{new_veh_mm4},  the problem of relay selection and power is solved using a centralized hierarchical deep reinforcement learning based method. Yet the authors us a simplified highway scenario, which would not scale for a city structure. 





Stochastic geometry studies have shown results on the interactions between vehicles on the highways or in the street intersections ~\cite{haenggi1, gall2019relayassisted}.
The work in \cite{MH_vehicular} performs statistical studies on traces of taxis to identify a planar point process that matches the random vehicle locations. The authors find that a Log  Gaussian Cox Process provides a good fit for particular traces.
In \cite{new_veh_mm5} propose a novel framework for performance analysis and design of relay selections in mmWave multi-hop V2V communications. More precisely, the distance between adjacent cars is modeled as shifted-exponential distribution.

Self-similarity for urban ad hoc networks has been introduced in \cite{spaswin, gsi}, where the hyperfractal model exploits the fractal features of urban ad hoc networks with road-side infrastructure. 
%
In~\cite{archiv17}, we presented an analysis of the propagation of information in a vehicular network where the cars (the only communication entities) are modeled using the hyperfractal model. As there are no relays in the intersections, as in the current study, in \cite{archiv17} we are in a disconnected network scenario where, as the nodes are allowed to move, the network becomes connected over time with mobility. The packets are being broadcast and results on typical metrics for delay tolerant networks are presented. There is no investigation on  power or energy. 
The study in \cite{spaswin} provides results on the minimal path routing using the hyperfractal model for static nodes to model the road-side infrastructure and assumes an infinite radio range. This is a concern for allowed transmission power and network energy consumption. In contrast to this first study, in this paper we add constraints on these quantities to provide insights on the achievable trade-offs between the end-to-end transmission energy and delay.

\section{System Model}\label{model}

In this section, we recall the necessary definitions of Hyperfractal model and enhance it to be able to formalise urban settings in matter of vehicles (as mobile users) as well as communication relays (as fixed infrastructures), both supported by a deterministic structure (called the support) on which various Poisson processes are sampled. 

\subsection{Hyperfractals for vehicular networks}
Cities are hierarchically organized \cite{Batty2008TheSS}. The main parts of the cities have many elements in common (in functional terms) and repeat themselves across several spatial scaling. This is reminiscent of a fractal, described by Lauwerier~\cite{LauwerierK} as an object that consists of an identical motif repeating itself on an ever-reduced scale.
The hyperfractal model has been introduced and exploited in ~\cite{spaswin,gsi} under static settings and in ~\cite{archiv17} under mobile settings. 



To represent the reality while being able to analyse various features, the map model lays in the unit square embedded in the 2-dimensional Euclidean space where various processes are sampled.
Basically,  we introduce a {\em support of the intensity measure} as a deterministic set (or structure) on which  respective Poisson processes are sampled (the set where this intensity measure is not null).
In this paper, the support of the population of $n$ nodes is a grid of streets. Let us denote this structure by $\mathcal{X}=\bigcup_{l=0}^\infty \mathcal{X}_l$ with 
$$\mathcal{X}_l:=\{(b2^{-(l+1)},y), b=1,3,\ldots, 2^{l+1}-1, y \in [0,1]\} 
\cup \{(x,b2^{-(l+1)}), b=1,3,\ldots,2^{l+1}-1, x \in [0,1]\},
$$
where $l$ denotes the level and $l$ starts from $0$, and $b$ is an odd integer.
Three first levels, $l=0,1,2$, are displayed in Figure~\ref{fig:map_support}. 
Observe that the central "cross" $\mathcal{X}_0$  splits $\bigcup_{l=1}^\infty \mathcal{X}_l$ in $4$ "quadrants" which all are  homothetic  to  $\mathcal{X}$ with the scaling factor $1/2$.

\subsection{Mobile users}
Following~\cite{spaswin},
we consider the Poisson point process $\Phi$ of (mobile) users on $\mathcal{X}$
with total intensity (mean number of points) $n$ ($0<n<\infty$)
having 1-dimensional intensity 
\beq
\lambda_l=n(p/2)(q/2)^l
\label{eq:dens_mobiles}
\eeq
on $\mathcal{X}_l$, $l=0,\ldots,\infty$,
with $q=1-p$ for some parameter $p$ ($0\le p\le 1$).
Note that $\Phi$ 
can be constructed in the following way: one samples the total  number of mobiles users $\Phi(\mathcal{X})=n$ from Poisson distribution; 
each mobile is placed independently 
with probability $p$ on $\mathcal{X}_0$ according to the uniform distribution
and with probability $q/4$ it is recursively located in a similar way in one the four quadrants of $\bigcup_{l=1}^\infty \mathcal{X}_l$.




\begin{figure}\centering
\begin{subfigure}{0.45\textwidth}
\vspace{0.2cm}\centering
\includegraphics[scale=0.22]{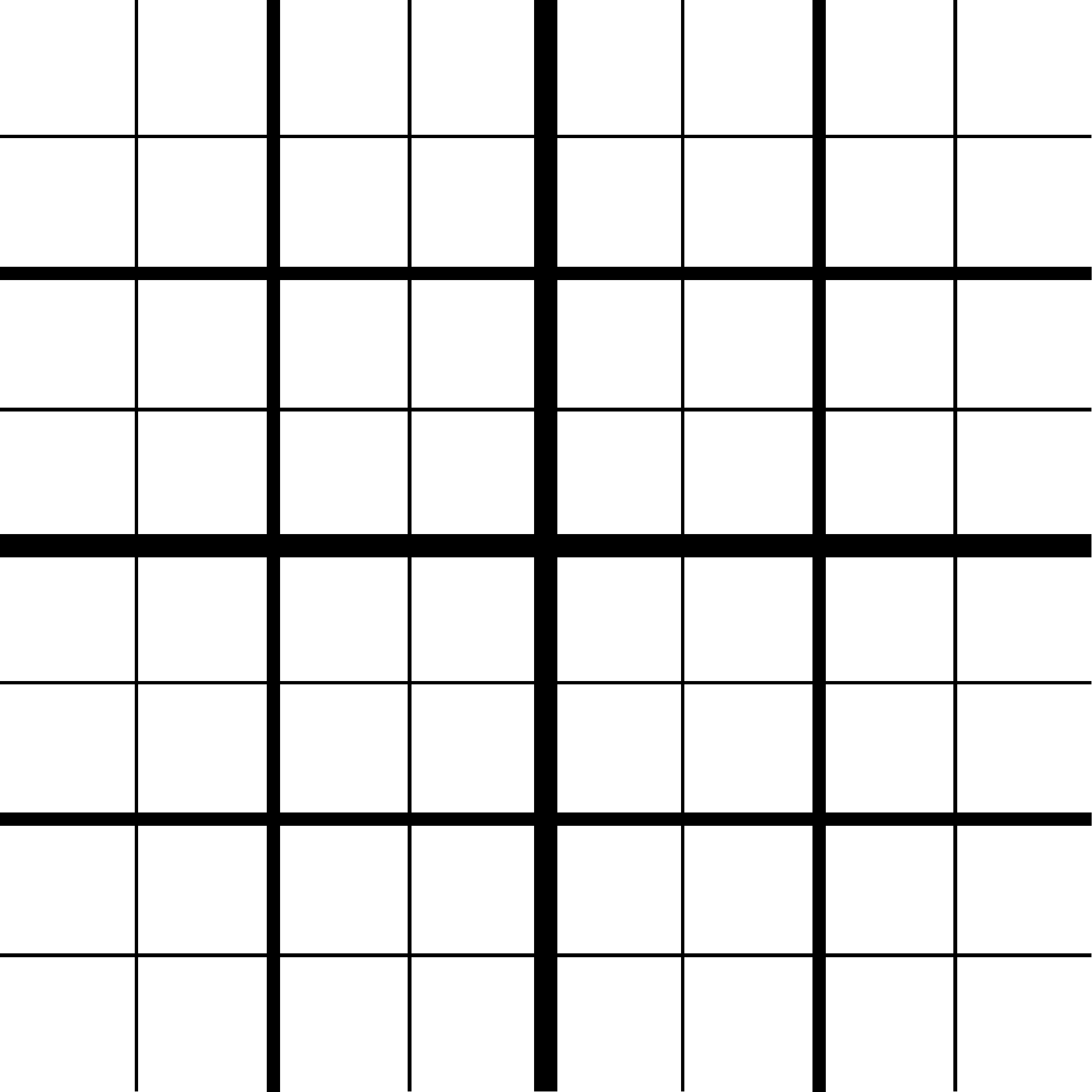}
\caption{}
\label{fig:map_support}
\end{subfigure}
\begin{subfigure}{0.45\textwidth}\centering
\includegraphics[scale=0.25, trim=5cm 1.5cm 0cm 2.5cm]{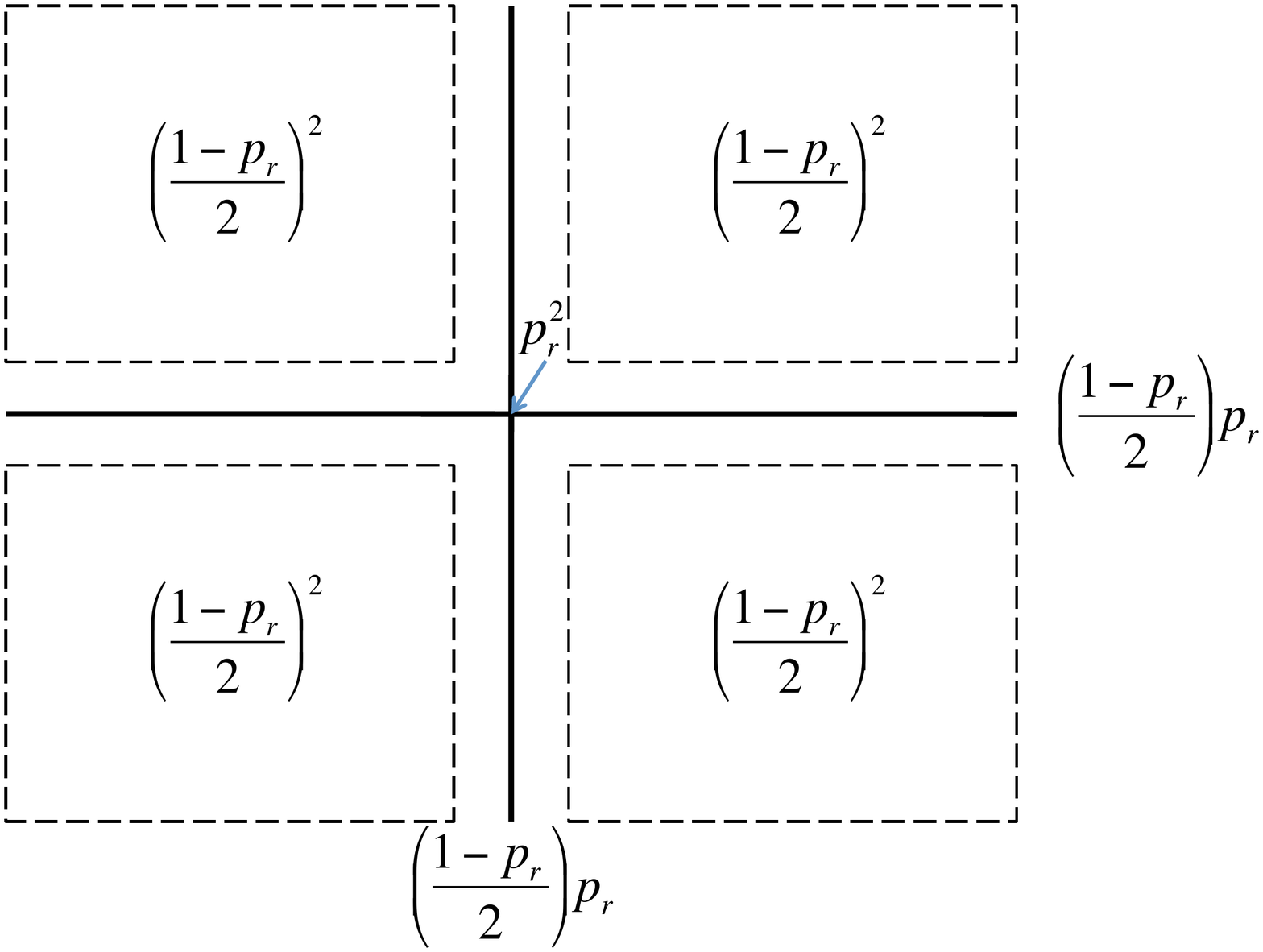}
\caption{}
\hspace{-5.2cm}
\vspace{-1.2cm}
\label{fig:relays}
\end{subfigure}
\caption{(a) Hyperfractal support; (b) Relays process construction}
\end{figure}

The process $\Phi$ is neither stationary nor  isotropic.
However, it has the following self-similarity property:
the intensity  measure of  $\Phi$ on  $\mathcal{X}$ is hypothetically reproduced in each of the four quadrants of $\bigcup_{l=1}^\infty \mathcal{X}_l$ with the scaling of its support by the factor 1/2 and of its value by $q/4$. 

The fractal dimension is a scalar parameter characterizing a geometric object with repetitive patterns. It indicates how the volume of the object decreases when submitted to an homothetic scaling. When the object is a convex subset of an euclidian space of finite dimension, the fractal dimension is equal to this dimension. When the object is a fractal subset of this euclidian space as defined in \cite{mandelbrot}, it is a possibly non integer but positive  scalar strictly smaller than the euclidian dimension. When the object is a measure defined in the euclidian space, as it is the case in this paper, then the fractal dimension can be strictly larger than the euclidian dimension. In this case we say that the measure is {\it hyper-fractal}. 
 \begin{remark}
 The fractal dimension $ d_F$ of the intensity measure of $\Phi$ satisfies
\begin{equation*} \label{eq:d_F}
\left(\frac{1}{2}\right)^{d_F}=\frac{q}{4} \qquad\text{thus}\qquad d_F=\frac{\log(\frac{4}{q})}{\log 2}\ge 2.
\end{equation*}
\end{remark}
The fractal dimension $d_F$ is greater than~2, the Euclidean dimension of the square in which it is embedded, thus 
the model was coined  {\em hyperfractal} in~\cite{spaswin}.
Notice that when $p=1$ the model reduces to the Poisson process on the central cross, while for $p\to 0$, $d_F\to 2$ it corresponds to  the uniform measure in the unit square.




\subsection{Relays}
Not surprisingly, the locations of communication infrastructures in urban settings also display a self-similar behavior: their placement is dependent on the traffic density. Hence we apply another hyperfractal process for selecting the intersections where a road-side relay is installed or the existing traffic light is used as road-side unit. This process has been introduced in~\cite{spaswin}. 

We denote the relay process by $\Xi$. To define $\Xi$ it is convenient to consider an auxiliary  Poisson process $\Phi_r$ with both processes supported by a 1-dimensional subset of $\mathcal{X}$ namely, the set of intersections of segments constituting $\mathcal{X}$. We assume that $\Phi_r$ has discrete intensity   
\begin{equation}\label{eq:intersection}
p(h,v)=\rho{p_r}^2\left(\frac{1-{p_r}}{2}\right)^{h+v}
\end{equation}
at all intersections $\mathcal{X}_h\cap\mathcal{X}_v$
for $h,v=0,\ldots,\infty$ for some parameter ${p_r}$, $0\le {p_r}\le 1$ and $\rho>0$. That is, 
at any such intersection the mass of $\Phi_r$ is Poisson random variable with parameter $p(h,v)$ and $\rho$
is the total expected number of points of $\Phi_r$ in the model.
The self-similar structure of $\Phi_r$ is explained by its construction:
we first sample the total number of points from a Poisson distribution
of intensity $\rho$ and given $\Phi_r(\mathcal{X})=M$,
each point is independently placed with probability ${p_r}^2$ in the central crossings of  $\mathcal{X}_0$, with  probability ${p_r}\frac{1-{p_r}}{2}$  on some other crossing of one  of the four segments forming  $\mathcal{X}_0$ and, with the remaining probability $\left(\frac{1-{p_r}}{2}\right)^2$, in a similar way, recursively, on some crossing of one of the four quadrants of $\bigcup_{l=1}^\infty \mathcal{X}_l$.
This is  illustrated in Figure~\ref{fig:relays}.
The Poisson process $\Phi_r$ is not simple:  
we define the relay process $\Xi$ as the support measure of $\Phi_r$, i.e., 
only one relay is installed at crossings where $\Phi_r$ has at least one point.

\begin{remark}
Note that the relay process $\Xi$ forms a non-homogeneous binomial point process (i.e. points are placed independently) on the crossings of $\mathcal{X}$ with a given intersection of two segments from $\mathcal{X}_h$ and $\mathcal{X}_v$ occupied by a relay point with probability $1-\exp(-\rho p(h,v))$.
\end{remark}
Similarly to the process of users, we can define the fractal dimension of the relay process.
\begin{remark}
The fractal dimension $d_r$ of the probability density of $\Xi$ is equal to the fractal dimension of the intensity measure of the Poisson process $\Phi_r$ and verifies 
\begin{equation*} \label{eq:d_r}
d_r=2\frac{\log(2/(1-{p_r}))}{\log 2}.
\end{equation*}
\end{remark}
A complete hyperfractal map with mobile nodes and relays is illustrated in Figure \ref{fig:hf_map}.

\begin{figure}[httb]\centering
\includegraphics[scale=0.55, trim=3cm 9.5cm 0cm 9.5cm]{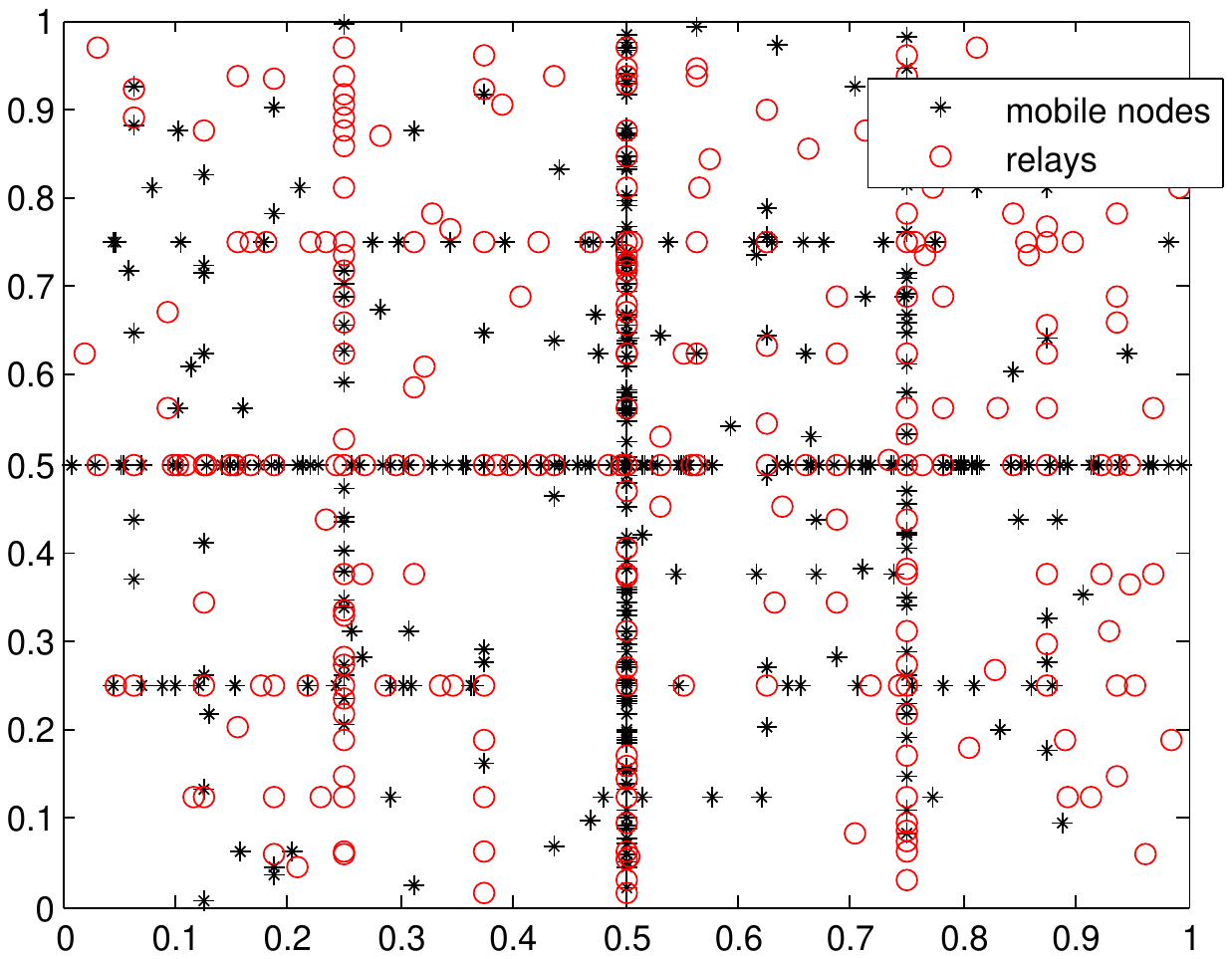}
\caption{Complete hyperfractal map with mobile nodes ("+") and relays ("o")}
\label{fig:hf_map}
\end{figure}

\section{Hyperfractal Properties and Communication model}\label{properties}

In this section we extract the relevant properties of the Hyperfractal model and relate them to our communication model. We also provide some additional insights into these models via the framework of the stochastic geometry and point process. These latter results are of independent interest and allow to lay foundations for other works.

\subsection{Number of relay nodes and asymptotic estimate}

We shall now provide the proof for the average number of relays in the map, this time, exploiting the "fractal-like" properties of our model and providing useful asymptotic estimates. 


\begin{theorem}\label{theo:relays}
The average total number of relays $R(\rho)$ in the map is:
 \begin{equation} 
 R(\rho)=\sum_{H,V}2^{H+V}\left(1-\exp(-\rho p(H,V))\right)=O(\rho^{2/d_r}\log\rho)
 \end{equation}
\end{theorem}
\begin{proof}
The probability that a crossing of two lines of level $H$ and $V$ is selected to host a relay is exactly $1-\exp(-\rho p(H,V))$. 

The average number of relays on a street of level $H$ is $L_H(\rho)$ and satisfies:
\begin{equation*}
L_H(\rho)=\sum_{V\ge 0} 2^V\left(1-\exp(-p(H,V)\rho)\right).
\end{equation*}
We notice that $L_H(\rho)=L_0((\frac{1-{p_r}}{2})^H\rho)$ and that $L_0(\rho)$ satisfies the functional equation:
\begin{equation*}
L_0(\rho)=1-\exp(-{p_r}^2\rho)+2L_0\left(\rho\frac{1-{p_r}}{2}\right).
\end{equation*}
It is known from \cite{mellin} and \cite{depois} that this classic equation has a solution such as $L_0(\rho)=O(\rho^{2/d_r})$. 

The average total number of relays $R(\rho)$ in the city is obtained by summing the average number of relays over all streets parallel to a given direction, {\it e.g.} the West-East direction (summing on all streets would count twice each relay). Since there are $2^H$ West-East streets at level $H$:
\begin{equation}
R(\rho)=\sum_H 2^HL_H(\rho)=\sum_{H,V\ge 0} 2^{H+V}\left(1-\exp(-p(H,V)\rho)\right)=\sum_{k=0}^\infty (k+1) 2^k(1-\exp(-\rho p_r^2((1-{p_r})/2)^{k})
\label{eq:relayk}\end{equation}
and satisfies the functional equation
\begin{equation*}
R(\rho)=L_0(\rho)+2R(\rho\frac{1-{p_r}}{2}).
\end{equation*}
From the same references (\cite{mellin,depois}), one gets
 \begin{equation*} \label{eq:nr_relays}
 R(\rho)=O(\rho^{2/d_r}\log\rho)
 \end{equation*}
Since $2/d_r<1$, the number of relays is much smaller than $\rho$ when $\rho\to\infty$. 
\end{proof}
Let us verify numerically the claim of Theorem~\ref{theo:relays}. 
 We generate the hyperfractal maps with several values of $n$, $\rho_n=n$ and $d_F=3$. Let us remind the reader that the theorem gives the expression~(\ref{eq:relayk}) of $R(\rho)$ as a sum with $k \to \infty$. In reality, a limited number of terms of the sums suffice to approximate $R(\rho)$ with acceptable accuracy. We denote by $k_{max}$ the number of terms used to compute in the sum.
Figure \ref{fig:valid1} shows that the computed values of the number of relays approach the measured value for $k_{max}=40$. The precision is further enhanced when $k_{max}=60$.

\begin{figure}[httb]\centering
\includegraphics[scale=0.42, trim=0cm 6.5cm 0cm 7.5cm]{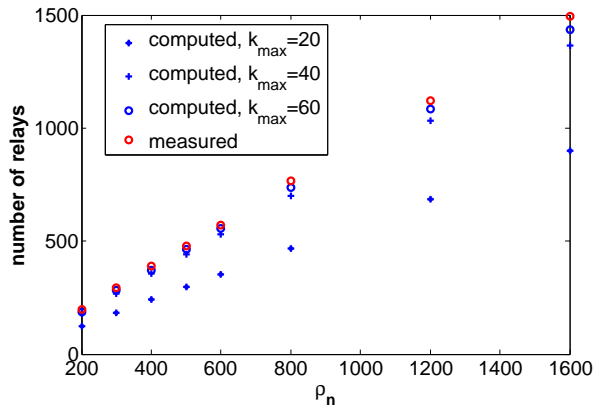}
\caption{Measured versus computed number of relays in the map for increasing values of $k_{max}$}
\label{fig:valid1}
\end{figure}

\subsection{Communication model}


As we primarily seek to understand the relationship between end-to-end communications and energy costs, we do not consider detailed aspects of the communication protocol that impact these (e.g., the distributed aspects needed to gather position information and construct routing tables in every node). 
The transmission is done in a half-duplex way, a node is not allowed to transmit and receive during the same time-slot. The received signal is affected by additive white Gaussian noise (AWGN) noise $N$ and path-loss with pathloss exponent $\delta\geq2$.


As a consequence of the high directivity and low permeability of the waves in high frequency (6GHz, 28GHz, 73 GHz as candidates for 5G NR), the next hop is always the next neighbor on a street, i.e. there exists no other node between the transmitter and the receiver. Indeed, while a lot of work is still dedicated to characterising the exact overall network connectivity for mmWave communications V2V in urban setting~\cite{Ozpolat2019}, it is known that intermediate vehicles create significant blockage and a severe attenuation of the received power for vehicles past near neighbours~\cite{NN_mmWave,Wimob19}. 
Thus the  routing strategy considered is a nearest neighbor routing.  In fact, we can show that, under reasonable assumptions, this strategy is optimal.


\begin{lemma}
If the noise conditions are the same around each node, then the nearest neighbor routing strategy is optimal in terms of energy.
\end{lemma}
\begin{proof}
To simplify this proof we ignore the signal attenuation due to the mobile users positioned as radio obstacles between the sender and the receiver of one hop packet transmission, although this will have an important impact on energy. Consider the packet transmission from a node at a location $x$ to a node at location $y$ on the same street. If $N$ is the noise level and $K$ is the required SNR, then the transmitter must use a signal of power $|y-x|^\delta NK$. Assume that there is a node at position $z$ between $x$ and $y$. Transmitting from $x$ to $z$ and then from $z$ to $y$ would require a cumulated energy $(|z-x|^\delta+|y-z|^\delta)NK$ which is smaller than the required energy for the direct transmission, since $|x-z|^\delta+|z-y
|^\delta\le \left(|x-z|+|z-y
|\right)^\delta$. 
\end{proof}



Let us make the simplifying assumption that all nodes on a street transmit with the same nominal power $P_m$ which depends only on the number $m$ of nodes on the street. We argue that a good approximation is to suppose that: 
\beq
P_m=\frac{\Pm}{m^\delta}
\label{eq:Pm}\eeq
where 
$\Pm$ is the transmitting power necessary for a node at one end of the street to transmit a packet directly to a node at the other end of the street. 
In other words, assume a road of infinite length where the nodes are regularly spaced by intervals of length $L$ is the length of our street. If in this configuration every node has a nominal power of $\Pm$, then the nominal power to achieve the same performance with a density $m$ times larger but with the same noise values should be $\Pm m^{-\delta}$ in order to cope with the loss effect. Thus would give expression~(\ref{eq:Pm}) if the nodes were regularly spaced by intervals of length $L/m$. But since the spacing intervals are irregular, one should cope with the largest gap $L_m/m$, this brings a small complication in the evaluation of $P_m$. But the probability that there exists a spacing larger than a given value $x/m$ is smaller than $m(1-\frac{x}{mL})^m\le m^{x/L}$. Thus we have $L_m=O(\log m/m)$ (asymptotically almost surely, and in fact as soon as $\lim\inf_m L_m/\log m>1$), and consequently $P_m=O(\Pm\log^\delta m/m^\delta)$. To help the reader, we focus on the expression~(\ref{eq:Pm}) as we are mainly interested in the order of magnitude.  

\begin{definition}
The \textit{end-to-end transmission delay} is represented by the total number of hops the packet takes in its path towards the destination. 
\end{definition}

As the energy to transmit a packet is the transmission power per unit of time, we consider the time necessary to send a packet as being equal to the length of a time-slot. We thus do not consider any MAC protocol for re-transmission and acknowledgment of the reception (e.g., we do not consider CSMA-like protocols). In any case, as it will be later observed throughout our derivations, varying the MAC protocol would just change some constants but not the overall scaling. Therefore, from now on, we will refer to $\Pm$ as the nominal power. 
Following this reasoning, the accumulated energy to cover a whole street containing $m$ nodes with uniform distribution via nearest neighbor routing is $mP_m=\frac{\Pm}{m^{\delta-1}}$. In this case, the larger the population of the street the smaller the nominal power and the smaller the energy to cover the street.

Relays stand in intersections, and thus on two streets with different values of $m$. 
We consider a relay to use two different radio interfaces, each with a transmission power according to the previously mentioned rule for each of the streets. This is a perfectly valid assumption, in line with 5G devices specifications for dual connectivity~\cite{3GPP}.
\subsection{Fundamental properties of the Poisson processes $\Phi$, $\Phi_r$, and $\Xi$}

In the following, we shall provide some fundamental tools that allow one to handle our model in a typical stochastic geometric framework. This section gives  insights about the theoretical foundations of hyperfractal point process which is of independent interest to our main results and can be used in other works. 

Let $L+1$ be a geometric random variable with parameter $p$ (i.e., $\mathbb{P}(L=l)=p(1-p)^{l}$, $l=0,1,\ldots$) and given $L$, let 
$x_0$ be the random  location uniformly chosen on $\mathcal{X}_L$. We call $x_0$ the {\em typical mobile user} of $\Phi$. More precisely, 
we shall consider the point process $\Phi\cup\{x_0\}$ where $x_0$ is sampled as described above and independently of $\Phi$.

Similarly, let $U+1$ and $W+1$ be two independent geometric random variables with parameter ${p_r}$ and given $(U,W)$, let 
$x_*$ be a crossing uniformly sampled from all the intersections of $\mathcal{X}_U\cap\mathcal{X}_W$.
 We call $x_*$ the {\em typical auxiliary point} of $\Phi_r$. More precisely, 
we shall consider point process $\Phi_r\cup\{x_*\}$ where $x_*$ is sampled as described above and independently of $\Phi_r$.

Finding the definition of the {\em typical relay node} $\xi_0$ is less explicit yet similar to the typical point definition.
Informally, the  conditional distribution of points ``seen'' from the origin given that the process has a point there is exactly the same as the conditional distribution of points of the process ``seen'' from an arbitrary location given the process has a point at that location. 

We define it as the random location on the set of the crossings of $\mathcal{X}$
involving the following biasing of the distribution of
$x_*$ by the inverse of the total number of the auxiliary points co-located with $x_*$
$$\mathbb{P}(\xi_0=x)=
\E\left[\frac{1(x_*=x)}{1+\Phi_r(\{x_*\})}\right]/\E\left[\frac{1}{1+\Phi_r(\{x_*\})}\right]\,.$$
More precisely we consider $\Xi'\cup\{\xi_0\}$ (which distribution is given
for any intersection $x$ of segments in $\mathcal{X}$ and a possible configuration~$\phi$ of relays) by considering
$$\mathbb{P}(\,\xi_0=x,\Xi'=\phi\,)=
\E\left[\frac{1(x_*=x)1(\text{supp}(\Phi_r)\setminus\{x\}=\phi)}{1+\Phi_r(\{x_*\})}\right]/
\E\left[\frac{1}{1+\Phi_r(\{x_*\})}\right]\,$$
where $\Phi_r$ and $x_*$ are independent (as defined above).
Note that,  in contrast to the typical points of  Poisson processes $\Phi$ and $\Phi_r$, the typical relay $\xi_0$ is not independent of remaining relays $\Xi'$.

In what follows, we shall prove that our typical points support 
the Campbell-Mecke formula (see~\cite{bartek_book, bartek_new_book}) 
thus justifying our definition and also providing an important tool for future exploiting the model in a typical stochastic geometric framework.

\begin{theorem}[Campbell-Mecke formula]\label{th.Campbell}
For all measurable functions $f(x,\phi)$ where $x\in\mathcal{X}$ and $\phi$ is a realization of a point process on $\mathcal{X}$
\begin{align}
\E\left[\sum_{x_i \in \Phi} f(x_i,\Phi)\right]&=n\E\left[f(x_0,\Phi\cup\{x_0\})\right]\label{e.CM-Phi}\,\\
\E\left[\sum_{x_i \in \Phi^r} f(x_i,\Phi_r)\right]&=\rho\E\left[f(x_*,\Phi_r\cup\{x_*\})\right]\label{e.CM-Phir}\\
\noalign{and}
\E\left[\sum_{x_i \in \Xi} f(x_i,\Xi)\right]&=\E\left[\Xi(\mathcal{X})\right]\E\left[f(\xi_0,\Xi'\cup\{\xi_0\})\right]\label{e.CM-Xi}\,
\end{align}
where the total expected number of relay nodes $\E\left[\Xi(\mathcal{X})\right]$ admits 
the following representation given by Theorem~\ref{theo:relays}
\begin{equation}\label{e:Xi-intenisty}
\E[\Xi(\mathcal{X})]=R(\rho)=\sum_{k=0}^\infty (k+1) 2^k(1-\exp(-\rho p_r^2((1-{p_r})/2)^{k}).
\end{equation}
\end{theorem}

\begin{proof}[Proof of Theorem~\ref{th.Campbell}]
First, consider the process of users $\Phi$.
The Campbell-Mecke formula and the Slivnyak theorem \cite{bartek_lecture_notes} for the  non-stationary Poisson point processes $\Phi$ give
\begin{equation}\label{eq:ns_c}
\E\left[\sum_{x_i \in \Phi} f(x_i,\Phi)\right]=\int_\mathcal{X} \E\left[f(x,\Phi\cup\{x\})\right]\mu(dx),
\end{equation}
where $\mu(dx)$ is the intensity measure of the process $\Phi$. 
Specifying this intensity measures the right-hand side term of~\eqref{eq:ns_c}, thus this becomes
\begin{equation*}
\sum_{l=0}^\infty \int_{\mathcal{X}_l} \E\left[f(x,\Phi \cup \{x\})\right] n(1-p)^lp dx.
\end{equation*}
In the above expression, one can recognize $\E\left[f(x_0,\Phi\cup\{x_0\})\right]$ which concludes the proof of~\eqref {e.CM-Phi}.
The proof of~\eqref{e.CM-Phir}
 follows the same lines.
 Consider now the relay process $\Xi$.
By the definition of~$\Xi$, one can express the left-hand side of~\eqref{e.CM-Xi} in the following way:
\begin{equation*}
\E\left[\sum_{x_i \in\Xi} f(x_i,\Xi)\right]=\E\left[\sum_{x_i \in\Phi_r}\frac{f(x_i,\text{supp}(\Phi_r))}{\Phi_r(\{x_i\})}\right]\,,
\end{equation*}
where $\text{supp}(\Phi_r)$ denotes the support of $\Phi_r$.
Using~\eqref{e.CM-Phir}, we thus obtain:
\begin{equation}\label{eq:Xi-Campblell2}
\E\left[\sum_{x_i \in \Xi} f(x_i,\Xi)\right]=\rho\E\left[\frac{f(x_*,\text{supp}(\Phi_r\cup\{x_*\}))}{1+\Phi_r(\{x_*\})}\right].
\end{equation}
By the definition of the joint distribution of  $x_*$ and $\text{supp}(\Phi_r\cup\{x_*\})$ 
the right-hand side of~\eqref{eq:Xi-Campblell2} is equal to 
$$
\rho\E\left[\frac{1}{1+\Phi_r(\{x_*\})}\right]\E\left[f(\xi_0,\Xi'\cup\{\xi_0\})\right]\,.$$
This completes the proof of~\eqref{e.CM-Phir} with 
$$\E\left[\Xi(\mathcal{X})\right]=\rho\E\left[\frac{1}{1+\Phi_r(\{x_*\})}\right].$$
\end{proof}

\section{Main Results}\label{results}
We now provide our theoretical bounds for the end-to-end communication hop count. The number of mobile nodes is exactly $n$, where  $n$ is an integer which  runs to infinity.
\subsection{Energy vs Delay}
Given that the transmitting power is dependent on the average density of the nodes on the streets and that the transmission power per node is limited by the protocols to a value of $\Pm$, the connectivity is restricted.
We introduce the following notions and notations.
Let $t$ be a node and let $P(t)$ be the nominal transmission energy of this node. 

\begin{definition}
Let $\setT$ be a sequence of nodes that constitutes a routing path. The path length is $D(\setT)=|\setT|$. The relevant energy quantities related to the paths are: 
\begin{itemize}
\item The path accumulated energy is the quantity $C(\setT)=\sum_{t \in \setT} P(t)$. 
\item The path maximum energy is the quantity $M(\setT) =\max_{t \in \setT} P(t)$. 
\end{itemize}
\end{definition}

The path accumulated energy is of interest as we want to optimize the quantity of energy expended in the-end-to-end communication, and respectively, the path maximum power as we want to find the path which maximum power does not exceed a given threshold depending on the energy sustainability of the nodes or the protocol. For example, it is unlikely that a node can sustain a nominal power of $\Pm$ equal to the power needed to transmit in a range corresponding to the entire length of a street. In this case it is necessary to find a path that uses streets with enough population to reduce the node nominal power and communication range (due to the mmWave technology limitations).

\begin{definition}\label{def_g_k}
\ 
\begin{itemize}
\item Let $G(n, \E)$ be the set of all nodes connected to the central cross with a path accumulated energy not exceeding $\E$. 
\item Let $G_k(n,\E)$ be the subset of $G(n, \E)$, where the path to the central cross should not go through more than $k$ fixed relays.
\end{itemize}
\end{definition}

\begin{definition}\label{def_g_prime_k}
Let $G'(n,\E)$ and $G_k'(n,\E)$ be the respective equivalents of $G(n,\E)$ and $G_k(n,\E)$ but with the consideration of the path maximum power instead of accumulated energy. 
\end{definition}

\subsection{Path accumulated energy}

The following theorem gives the asymptotic connectivity properties of the hyperfractal in function of the accumulated energy and in function of the path maximum power. This shows that for $n$ large, even for some sequences of energy thresholds $\E_n$ tending to zero, the sets $G_1(n,\E_n)$ asymptotically dominate the network. The same holds for the sets sequence $G'_1(n,\E_n)$.
\begin{theorem}\label{theo_g}
In an urban network with $n$ mobile nodes following a hyperfractal distribution, the following holds:
\beq
\lim_{n \to \infty}  \expect {\frac{|G_1(n, n^{-\gamma}\Pm)|}{n}}=1
\eeq
for $\gamma<\delta-1$ 

and
\beq
\lim_{n \to \infty}  \expect {\frac{|G'_1(n, n^{-\gamma}\Pm)|}{n}}=1
\eeq
for $\gamma<\delta$

where $\delta$ is the pathloss coefficient.
\end{theorem}

The following lemma ensures the existence of nodes in a street (with proof in the Appendix).
\begin{lemma}\label{lemma:existence}
There exists $a>0$ such that, for all integers $H$ and $n$, the probability that a street of level $H$ contains less than $n\lambda_H/2$ nodes or more than $2n\lambda_H$ nodes is smaller than $\exp(-an\lambda_H)$.
\end{lemma}

The following corollary gives a result on the scaling of the number of nodes in a segment of street and the accumulated energy, getting us one step closer to the results we are looking for. 
\begin{corollary}
Let $0<\phi\le 1$, assume an interval corresponding to a fraction $\phi$ of the street length. If the interval is on a street of level $H$, the probability that it contains less than $\phi\lambda_Hn/2$ nodes and it is covered with accumulated energy greater than $\phi (n\lambda_H)^{1-\delta}\Pm$ is smaller than $e^{-a \phi\lambda_Hn}$.
\end{corollary}
\begin{proof}
This is a slight variation of the previous proof.
If we denote by $N_H(n,\phi)$ the number of nodes on the segment, we have
$\mathbb{E}[e^{tN_H(n,\phi)}]=\left(1+\lambda_H \phi(e^z-1)\right)^n$. The previous proof applies by replacing $\lambda_H$ by $\phi\lambda_H$.
The accumulated energy has the expression $\Pm\frac{N_H(n,\phi)}{N_H^\delta(n)}$. Further applying the previous reasoning to each of the random variables $N_H(n)$ and $N_H(n,\phi)$ gets the result. 
\end{proof}

Throughout the rest of the paper, we only consider the cases where $d_F>3$ and $d_r<d_F-1$, {\it i.e.} $(2/q)^2<2/(1-p_r)$.

The following theorem is the \textbf{main result} of our paper and shows that increasing the path length decreases the accumulated energy. In fact, for $n \to \infty$, the limiting energy goes to zero. 

\begin{theorem}\label{theo_energy1}
In a hyperfractal with $n$ nodes, with nodes of fractal dimension $d_F$ and relays of fractal dimension $d_r$, 
the shortest path of accumulated energy $E_n=c_E n^{(1-\delta)(1-\alpha)}\Pm$, where $c_E>0$ and $\alpha<1$, between two nodes belonging to the giant component $G_1(n,E_n)$, passes through a number of hops : 
\beq \label{eq:energy1}
D_n=O(n^{1-\alpha/(d_F-1)}) 
\eeq
\end{theorem}

Although the source and the destination belong to $G_1(n,E_n)$, it is not necessary that all the nodes constituting the path also belong to $G_1(n,E_n)$, i.e., the path may include nodes that are more than one hop from the central cross. 
\begin{remark}
We have the identity 
\begin{equation}
\left(\frac{E_n}{\Pm}\right)^{1/(\delta-1)}D_n^{d_F-1}=O(n^{d_F-2}).
\end{equation}
\end{remark}

Let us now prove the theorem.
\begin{proof}
The main part of our proof is to consider the case when the source, denoted by $m_H$, and the destination, $m_V$, both stand on two different segments of the central cross. In this case, we consider the energy constraint $\frac{1}{3}E_n$. We can easily extend the result to the case when the source and the destination stand anywhere in the giant component $G_1(n,E_n)$ by taking $E_n$ as energy constraint and the theorem follows. 

When $m_H$ and $m_V$ are on the central cross, there exists a direct path that takes the direct route by staying on the central cross, more specifically, in Figure~\ref{fig:routing}~a), the segments $[SA]$,$[AO]$,$[OC]$,$[CD]$. Then, the path length is of order $\Theta(n)$ while the accumulated energy of order $\Theta(n^{1-\delta})\Pm$.

In order to significantly reduce the order of magnitude of the path hop length, one must consider a diverted path with three fixed relays, as indicated in Figure~\ref{fig:routing}~a). The diverted path proceeds into two streets of level $x$. Let $\setT$ be the path. It is considered that $x=\alpha\frac{\log n}{\log(2/(1-p_r))}$ for $\alpha<1$. 
\begin{figure}\centering
\includegraphics[scale=0.4,  trim={3cm 4cm 4cm 2cm}]{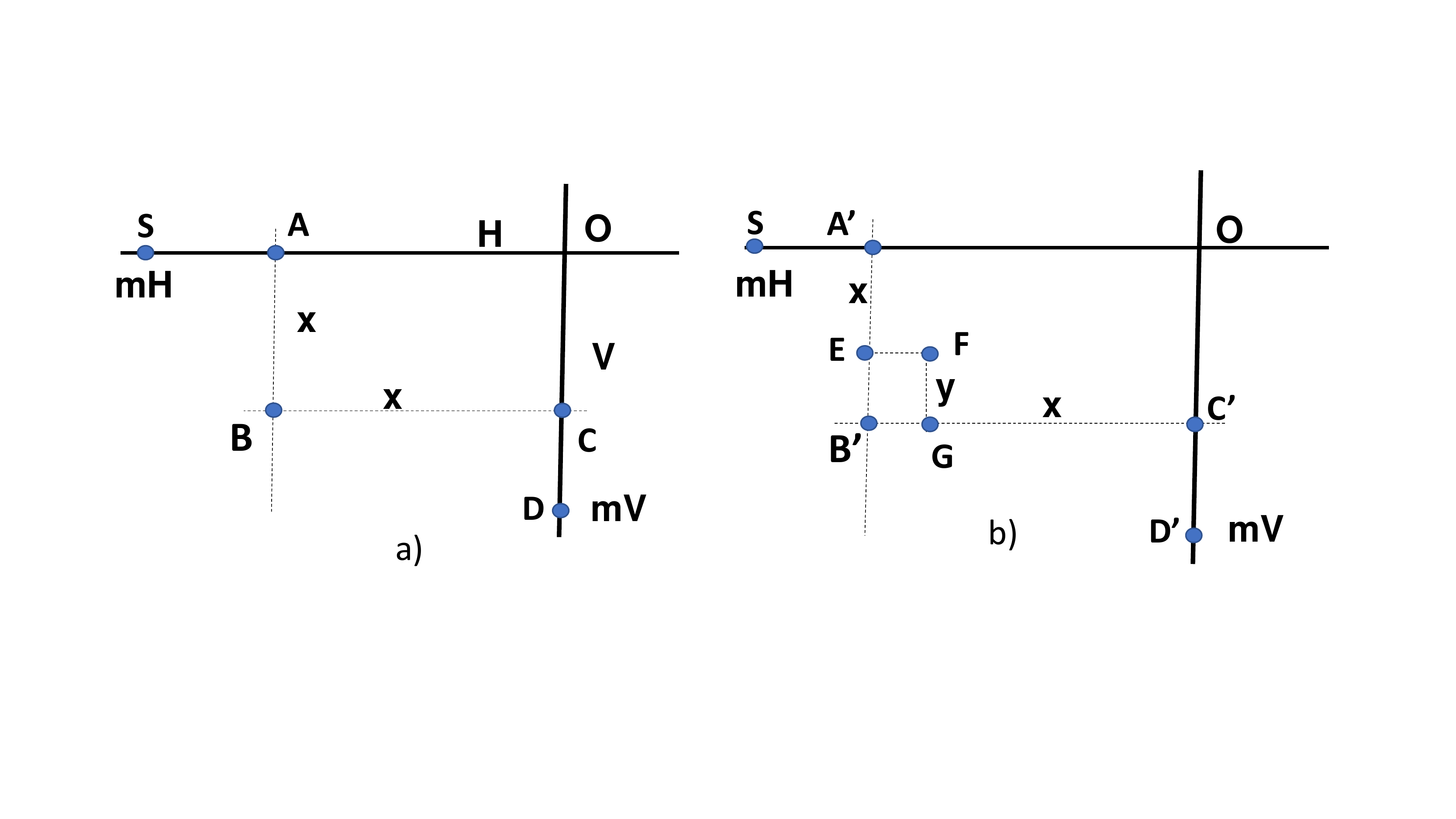}
\caption{a) Diverted path with three fixed relays (left), b) five fixed relays (right).
} \label{fig:routing}
\vspace{-0.2cm}
\end{figure}
%
The path is made of two times two segments: the segment of street $[SA]$ on the central cross which corresponds to the distance from the source to the first fixed relay to a street of level $x$,  and then the segment $[AB]$ between this relay and the fixed relay to the crossing street of level $x$. The second part of the path is symmetric and corresponds to the connection between this relay and the destination through segment $[BC]$ and $[CD]$. 

Denote by $L(x,y)$ the distance from an arbitrary position on a street of level $y$ to the first fixed relay to a street of level $x$. The probability that a fixed relay exists at a crossing of two streets of respective level $x$ and $y$ is $1-\exp(-\rho_n p(x,y))$. Since the spacing between the streets of level $x$ is $2^{-x}$, it is known from \cite{spaswin} that 
\begin{equation*}
L(x,y)\le\frac{2^{-x}}{1-\exp(-\rho_n p(x,y))}
\end{equation*}
where $\rho_n$ is the effective number of relays in the map (reminding that $\rho_n=n$ to simplify). 
The probability that the two streets of level $x$ have a fixed relay at their crossing is $1-\exp(-\rho_n p(x,x))$. With the condition $\rho=n$, one gets $\rho_n p(x,x)=n^{1-2\alpha\log(2/q)/\log(2/(1-{p_r}))}>n^{1-\alpha}$ since $2\log(2/q)<\log(2/(1-{p_r}))$. Therefore the probability that the relay does not exist decays exponentially fast. Since the accumulated energy of the path, $E(\setT)$, satisfies with high probability
\begin{equation*}
E(\setT)= O(L(x,0)n^{1-\delta}\Pm)+O((n\lambda_x)^{1-\delta}\Pm)
\end{equation*}
and the average number of nodes of the path, $D(\setT)$,  satisfies with probability tending to 1, exponentially fast:
\begin{equation*}
D(\setT)= O(L(x,0)n)+O(n\lambda_x).
\end{equation*}
Then, with the value $x=\alpha\frac{\log n}{\log(2/(1-{p_r}))}$, we detect that the main contributor of the accumulated energy are the segments $[AB]$ and $[BC]$, namely 
$E(\setT)=O\left(n^{(1-\delta)(1-\alpha)}\right).$
and let $c_E$ such that
$E(\setT)\le\frac{c_E}{3}n^{(1-\delta)(1-\alpha)}.$
The main contributor in hop count in the path is in fact in the parts which stands on the central cross, namely $[m_HA]$ and $[m_VC]$:
$D(\setT)=O\left(n^{1-\alpha/(d_F-1)}\right).$
\end{proof}

In Theorem~\ref{theo_energy1}, it is always assumed that $E_n\to 0$, since $\alpha<1$. In this case, $D_n$ spans from  $O(n^{1-1/(d_F-1)})$ to $O(n)$ (corresponding to a path staying on the central cross). When the fractal dimension $d_F$ is large it does not make a large span. In fact, if $E_n$ is assumed to be constant, {\it i.e.} $\alpha=1$, then we can have a substantial reduction in the number of hops, as described in the following theorem.


\begin{theorem}
In a hyperfractal with $n$ nodes, with nodes of fractal dimension $d_F$ and relays of fractal dimension $d_r$, the shortest path of accumulated energy $E_n=v_E\Pm$ with $v_E>6$,  between two nodes belonging to the giant component $G_1(n,E_n)$, passes through a number of hops :
\begin{equation*}
D_n = O\left(n^{1-\frac{2}{d_r(1+1/d_F)}}\right)
\end{equation*}

\end{theorem}

The theorem shows the achievable limits of number of hops when the constraint on the path energy is let loose. In fact, this allows taking the path with five fixed relays (Figure \ref{fig:routing}). The condition on $v_E>6$ comes from the 5 relays plus the step required to escape the giant component.  
\begin{remark}
When $d_r\to 2$ then $D_n=O( n^{1/(d_F+1)})$, and the hyperfractal model is behaving like a hypercube of dimension $d_F+1$. Notice that in this case $D_n$ tends to be $O(1)$ when $d_F\to\infty$. 
\end{remark}
\begin{proof}
In the proof of Theorem~\ref{theo_energy1}, it is assumed that $x<\frac{\log n}{\log(2/(1-{p_r}))}$ in order to ensure that the number of hops on the route of level $x$ tends to infinity. However, we can rise the parameter $x$ in the range $\frac{\log n}{\log(2/(1-{p_r}))}\le x <\frac{\log n}{2\log(2/q)}$.

We have $n\lambda_x\to 0$. In this case, $E(\setT)\to 2\Pm$ since the streets of level $x$ are empty of nodes with probability tending to 1. Let us denote $x=\beta\frac{\log n}{2\log(2/q)}$ with $\beta<1$. We have  $D(\setT)=O(L(x,0)n)=O(n^{1-\beta/d_r})$. Clearly, $\beta$ cannot be greater than 1 as, in this case, the two streets of level $x$ will not hold a fixed relay with high probability and the packet will not turn at the intersection. Therefore the smallest order that one can obtain on the diverted path with three relays is limited to $n^{1-1/d_r}$, which is not the claimed one.

To obtain the claimed order, one must use the diverted path with five fixed relays, as shown in figure~\ref{fig:routing}~b). The diverted path is composed by the segments: $[SA']$,$[A'E]$,$[EF]$,$[FG]$,$[GC']$ and $[C'D']$.
It is shown in~\cite{spaswin} that the order can be decreased to $n^{1-2/((1+1/d_F)d_r)}$.
\end{proof}

\subsection{Path maximum power}

The next results revisit the previous theorems on the path accumulated energy in the corresponding case of the imposed constraint on the path maximum power.
\begin{theorem}\label{theo_energy2}
The shortest path of maximum power less than $M_n=n^{-\delta(1-\alpha)}\Pm$ with $\alpha<1$, between two nodes belonging to the giant component $G_1'(n,M_n)$, passes through a number of hops:
$$
D_n=O\left(n^{1-\alpha/(d_F-1)}\right)
$$
\end{theorem}

It is \textbf{important to note} that although the orders of magnitude of path length $D_n$ are the same in both Theorem~\ref{theo_energy1} and Theorem~\ref{theo_energy2}, the results consider two different giant components: (accumulated)  $G_1(n,E_n)$ and (maximum) $G'_1(n,M_n)$. 
\begin{remark}
We have the identity
\begin{equation}
\left(\frac{M_n}{\Pm}\right)^{1/\delta} D_n^{d_F-1}=O(n^{d_F-1-\delta}).
\end{equation}
\end{remark}
\begin{theorem}
Let the maximum path transmitting power between two points belonging to the giant component, $G_1'(n,M_n)$ be $M_n=\Pm$. The number of hops $D_n$ on the shortest path  is $O\left(n^{1-2/(d_r(1+1/d_F)}\right)$.
\end{theorem}
This theorem gives the path length when no constraint on transmitting power exists (the maximum transmitting power allowed is the highest power for a transmission between two neighbors in the hyperfractal map). We obtain here the same results of~\cite{spaswin}, where an infinite radio range is considered, which is not a feasible result for mmWave technology deployments.


\subsection{Remarks on the network throughput capacity}
Let us consider the scaling of the network throughput capacity with constraints on the energy.
In \cite{capacity_formula}, the authors express the throughput capacity of random wireless networks as:    
 \beq \label{eqn:capacity}
  \zeta(n)=\Theta  \left({n^2\sum\nolimits_{i\in G}\omega_{i}(n)\over \sum\nolimits_{i,j\in G} r_{ij}}\right).
    \eeq
where $\zeta(n)$ is the throughput capacity, defined as the expected number of packets delivered to their destinations per slot, $\omega_{i}(n)$ is the expected transmission rate of each node $i$ among all the nodes $n$ and $G$ is the giant component. 
In the following, denote by $C$ the transmission rate of each node. 


Using our results of Theorems~\ref{theo_energy1} and~\ref{theo_energy2} and substituting them in (\ref{eqn:capacity}), we obtain the following corollary on a lower bound of the network throughput capacity with constraints on path energy. 
\begin{corollary}
In a hyperfractal with $n$ nodes, fractal dimension of nodes $d_F$, and $\alpha<1$ and $C$ the transmission rate of each node when either 
\begin{itemize}
\item $E_n=O\left(n^{(1-\delta)(1-\alpha)}\Pm\right)$ is the maximum accumulated energy of the minimal path between any pair of nodes in the giant component $G_1(n,E_n)$ 
\end{itemize}
or
\begin{itemize}
\item $M_n=O(n^{-\delta(1-\alpha)}\Pm)$ is the maximum path power of the minimal path between any pair of nodes in the giant component $G_1'(n,M_n)$,
\end{itemize}
a lower bound on the network throughput is:
\beq
\zeta(n)=\Omega\left(C n^{\frac{\alpha}{d_F-1}} \right)
\eeq

\end{corollary}

\begin{remark}
We notice that with $\alpha<1$ and $d_F>3$ we have $\zeta(n)$ of order which can be smaller than $n^{1/2}$ which is less than the capacity in a random uniform network with omni-directional propagation as described in~\cite{GuptaKumar}.
\end{remark}
\begin{remark}
When $\alpha=1$, {\it i.e.} with no energy constraint $E_n=c_E\Pm$ the path length can drop down to $D_n=O\left(n^{1-2/((1+1/d_F)d_r)}\right)$ and, in this case, we have $\zeta(n)=\Omega(n^{2/((1+1/d_F)d_r)})$ which tends to be in $O(n)$ when $d_F\to\infty$ and $d_r\to 2$. In this situation the capacity is of optimal order since $D_n$ tends to be $O(1)$.
\end{remark}

\section{Fitting the hyperfractal model for mobile nodes and relays to data }\label{calcul_dr}

The hyperfractal models for mobile nodes and for relays have been derived by making observations on the scaling of traffic densities and the scaling of the infrastructures, with road lengths,  distances between intersections which allow rerouting of packets, {\it etc}. 
Let us emphasize again that in the definition of the hyperfractal model, there is neither an assumption nor a condition on geometric properties (in the sense of geometric shape, strait lines, intersection angles, {\it etc}). Our description of a hyperfractal starting from the support $\mathcal{X}_0$ splitting the space into four quadrants is an example, (e.g., split by three to fit a Koch snowflake). 

In our previous works \cite{archiv17}, we have introduced a procedure which allows transforming traffic flow maps into hyperfractal by computing the fractal dimension $d_F$ of each traffic flow map then quantify the metrics of interest. We shall revisit the theoretical foundation of this procedure in order to compute the fractal dimension $d_r$ of the relay placement since as observed from Section \ref{model}, the placement of relays is dependent on the density of mobile nodes. 

\subsection{Theoretical Foundation: computation of the fractal dimension of the relays}

In this work, we state that the relaying infrastructure placement also follows a distribution with parameter of a fractal dimension, $d_r$. We thus now present a procedure for the computation of the fractal dimension of the road-side infrastructure in a city. 



The criteria for computing the fractal dimension of the road-side infrastructure are similar to the criteria used for computing the fractal dimension of the traffic flow map. 
%
%
The fitting procedure exploits the scaling between the length of different levels of the support $\mathcal{X}_l$ and the scaling  of the 1-dimensional intensity per level, $\lambda_l$. The difficulty is that the roads rarely have an explicit level hierarchy since the data we have about cities are in general about road segment lengths and average mobile nodes densities. 
To circumvent this problem, we do a ranking of the road segments in the decreasing order of their mobile density. If $\CS$ is a segment we denote $\eta(\CS)$ its density and $\CL(\CS)$ the accumulated length of the segment ranked before $\CS$ ({\it i.e.} of larger density than $\eta(\CS)$). 
For $\xi>0$ we denote $\mu(\xi)=\eta(\CL^{-1}(\xi))$. Formally $\CL^{-1}(\xi)$ is the road segment $S$ with the smallest density such that $\CL(\CS)\le\xi$. The hyperfractal dimension will appear in the asymptotic estimate of $\mu(\xi)$ when $\xi\to\infty$ via the following property:

\beq \label{eq:eq_fitting_mobiles}
\mu(\xi)=\Theta\left(\xi^{1-d_F}\right).
\eeq

The \textbf{procedure for the computation of the fractal dimension of the relays} is similar to the fitting procedure for mobile nodes, \cite{archiv17} and has the following steps. First, we consider the set of road intersection $\CI$ defined by the pair of segments $(\CS_1,\CS_2)$ such that $\CS_1$ and $\CS_2$ intersect. Let $\xi_1$, $\xi_2$ be two real numbers we define $p(\xi_1,\xi_2)$ as the probability that two intersecting segments $\CS_1$ and $\CS_2$ such that $C_l(\CS_1)\le\xi_1$ and $C_l(\CS_2)\le\xi_2$ contains a relay. The hyperfractality of the distribution of the relay distribution implies when $\xi_1,\xi_2\to\infty$:

\beq\label{eq:eq_fitting_relays}
p(\xi_1,\xi_2)=\Theta\left((\xi_1\xi_2)^{-d_r/2}\right).
\eeq
Since the probability is not directly measurable we have to estimate it via measurable samples. Indeed let $N(\xi_1,\xi_2)$ be the number of intersections $(\CS_1,\CS_2)\in\CI$ such that $C_l(\CS_1)\le\xi_1$ and $C_l(\CS_2)\le\xi_2$ and let $R(\xi_1,\xi_2)$ be the number of relays between segments $(\CS_1,\CS_2)$ such that $C_l(\CS_1)\le\xi_1$ and $C_l(\CS_2)\le\xi_2$. One should have: 

\beq\label{eq:eq_fitting_relays2}
\frac{R(\xi_1,\xi_2)}{N(\xi_1,\xi_2)}=\Theta\left((\xi_1\xi_2)^{-d_r/2}\right)
\eeq
 and from here get the fractal dimension of the relay process.
 
\subsection{Data Fitting Examples}
Using public measurements \cite{dataAdelaide_relays}, we show that the data validates the hyperfractal scaling of relays repartition with density and length of streets. While traffic data is becoming accessible, the exact length of each street  is  difficult to find, therefore the fitting has been done manually. 

Figure \ref{fig:traffic_lights1} shows a snapshot of the traffic lights locations in a neighborhood of Adelaide, together with traffic densities on the streets, when available.
As the roadside infrastructure for V2X communications has not been deployed yet or not at a city scale, we will use traffic light data as an example for RSU location. It is acceptable to assume the RSUs will have similar placement, as, themselves, traffic lights are considered for the location of the RSUs on having radio devices mounted on top of them \cite{traffic_RSU}. 
\begin{figure}\centering
\begin{subfigure}{0.45\textwidth}
\includegraphics[scale=0.3] {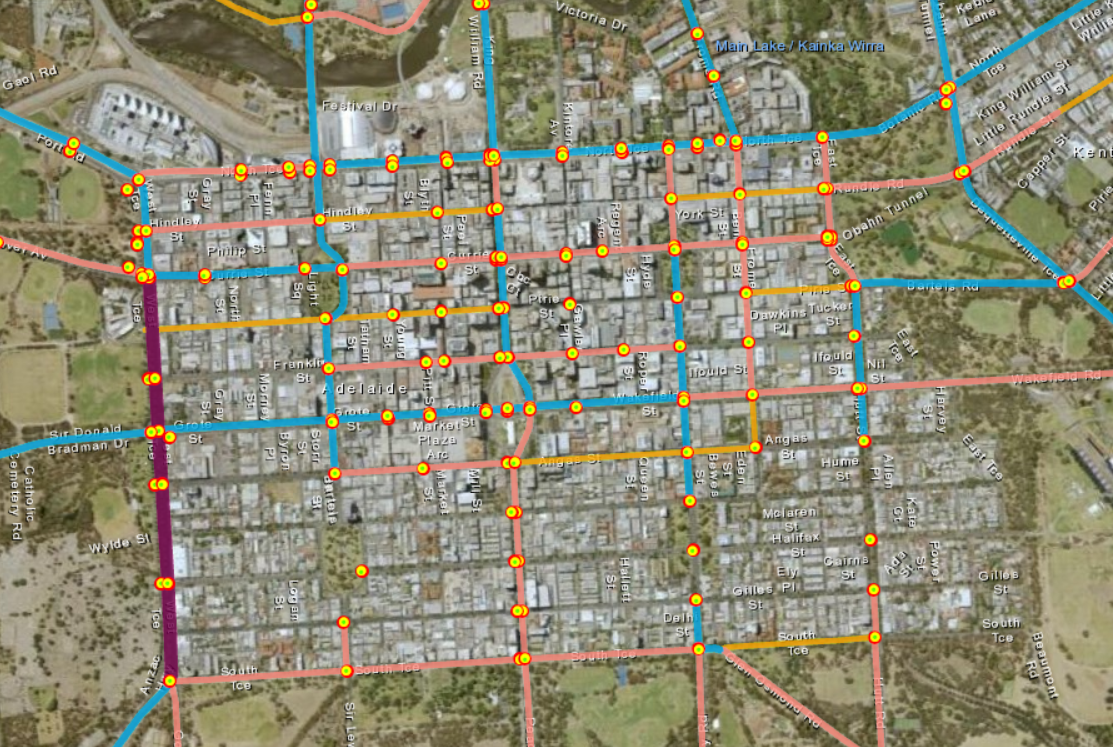}
\caption{Traffic lights data in Adelaide.} \label{fig:traffic_lights1}
\vspace{-0.2cm}
\end{subfigure}
\begin{subfigure}{0.45\textwidth}
\includegraphics[scale=0.45,  trim={1cm 9cm 0cm 10cm}]{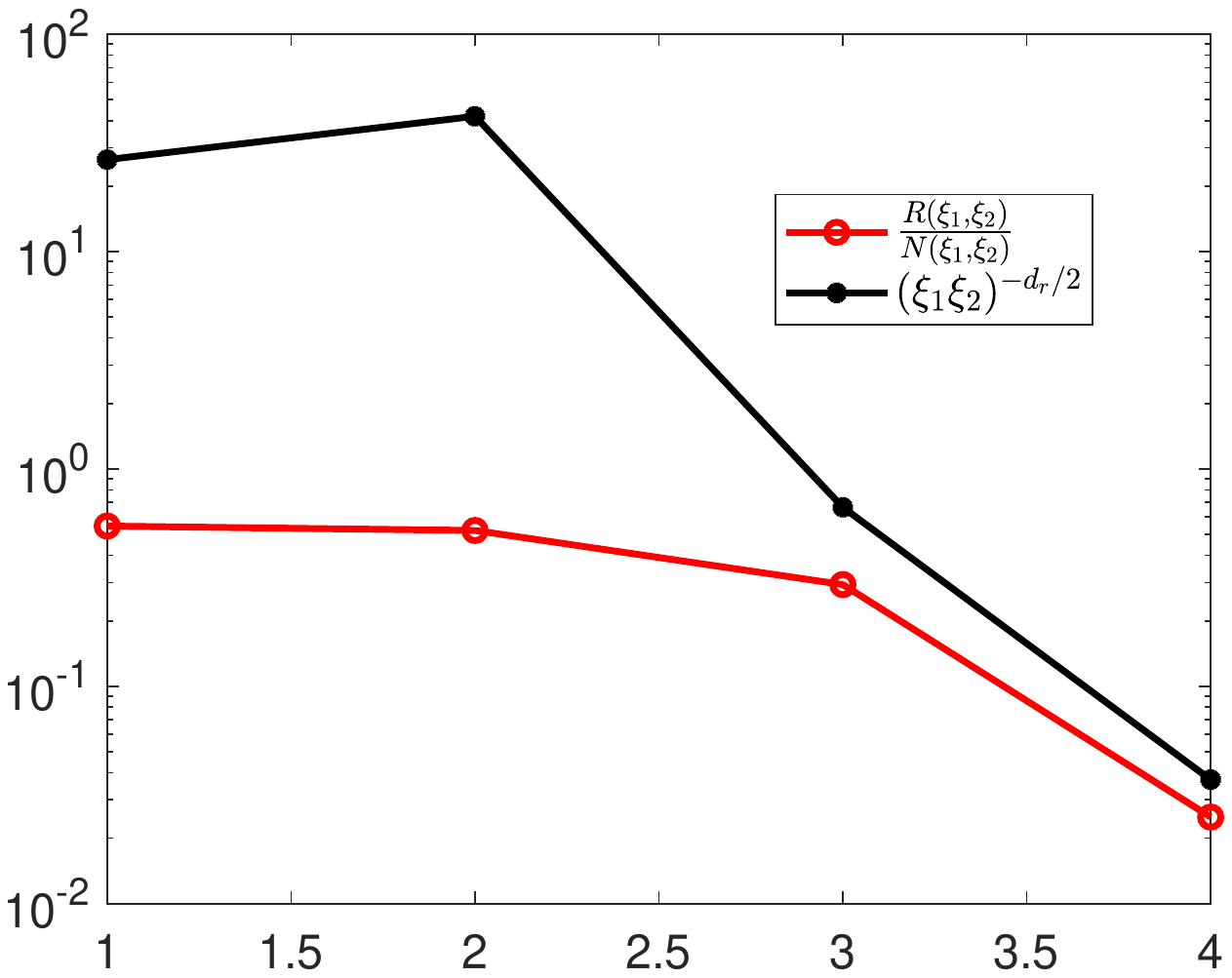}
\caption{\centering Computation of $d_r$}
\label{fig:fitting1}
\end{subfigure}
\caption{Data fitting for Adelaide}
\end{figure}
 By applying the described fitting procedure and using equation (\ref{eq:eq_fitting_relays2}) the estimated fractal dimension of the traffic lights distribution in Adelaide is $d_r=3.5$. 
In Figures \ref{fig:fitting1} we show the fitting of the data for the density repartition function.  

Note that it is the asymptotic behavior of the plots that are of interest (i.e., the increasing accumulated distance with decreasing density therefore decreasing the probability of having a relay installed) since the scaling property comes from the roads with low density, thus the convergence towards the rightmost part of the plot is of interest.


\section{Numerical Evaluation}\label{simulations}
We evaluate the accuracy of the theoretical findings in different scenarios by comparing them to results obtained by simulating the events in a two-dimensional network. We developed a MatLab discrete time event-based simulator following the model presented in Section~\ref{model}. 
\begin{figure}[tb]\centering
\includegraphics[scale=0.4,  trim={1cm 9.3cm 0cm 8.7cm}] {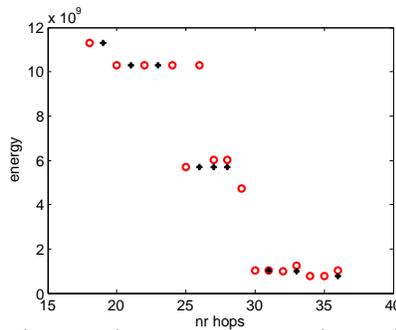}
\caption{Minimum accumulated end-to-end energy  versus hops  for a transmitter-receiver pair (fixed and allowed number of hops in red circles, and maximum number of hops in black stars).} \label{fig:energy_pair}
\vspace{-0.2cm}
\end{figure}
The length of the map is 1000 and, therefore, $\Pm$ is just $1000^\delta$, where $\delta$ is the pathloss coefficient that will be chosen to be $2,3$ or $4$, in line with millimeter-wave propagation characteristics.
Figure~\ref{fig:energy_pair} shows the trade-offs between accumulated end-to-end energy and hop count for a transmitter-receiver pair by selecting randomly pairs of vehicles in a hyperfractal map with $n=800$, pathloss coefficient $\delta=4$, fractal dimension of nodes $d_F=4.33$ and fractal dimension of relays $d_r=3$. The plot shows the minimum accumulated energy for the end-to-end transmission for a fixed and allowed number of hops, $k$, in red circle markers. Note that the energy does not decrease monotonically as forcing to take a longer path may not allow to take the best path. However when considering the minimum accumulated energy of all paths {\em up to a number of hops}, the black star markers in Figure~\ref{fig:energy_pair}, the energy decreases and exhibits the behavior claimed in Theorem~\ref{theo_energy1}. That is, the minimum accumulated energy is indeed decreasing when the number of hops is allowed to grow (and the end-to-end communication is allowed to choose longer, yet cheaper, paths).

Let us further validate Theorem~\ref{theo_energy1} through simulations performed for 100 randomly chosen transmitter-receiver pairs in hyperfractal maps with various configurations. We run simulations for different values of the number of nodes,  $n=800$ nodes and $1000$ nodes respectively, different values of pathloss, $\delta=2$ and $\delta=3$ and different configurations of the hyperfractal map. The setups of the hyperfractal maps are: node fractal dimension $d_F=4.33$ and relay fractal dimension $d_r=3.3$ for the first setup and $d_F=3.3$ and $d_r=2.3$ for the second setup.  

\begin{figure*}[h!]
\begin{subfigure}[t]{0.2\textwidth} \label{800_linear}
\includegraphics[scale=0.3, trim=5.5cm 8cm 0cm 9cm]{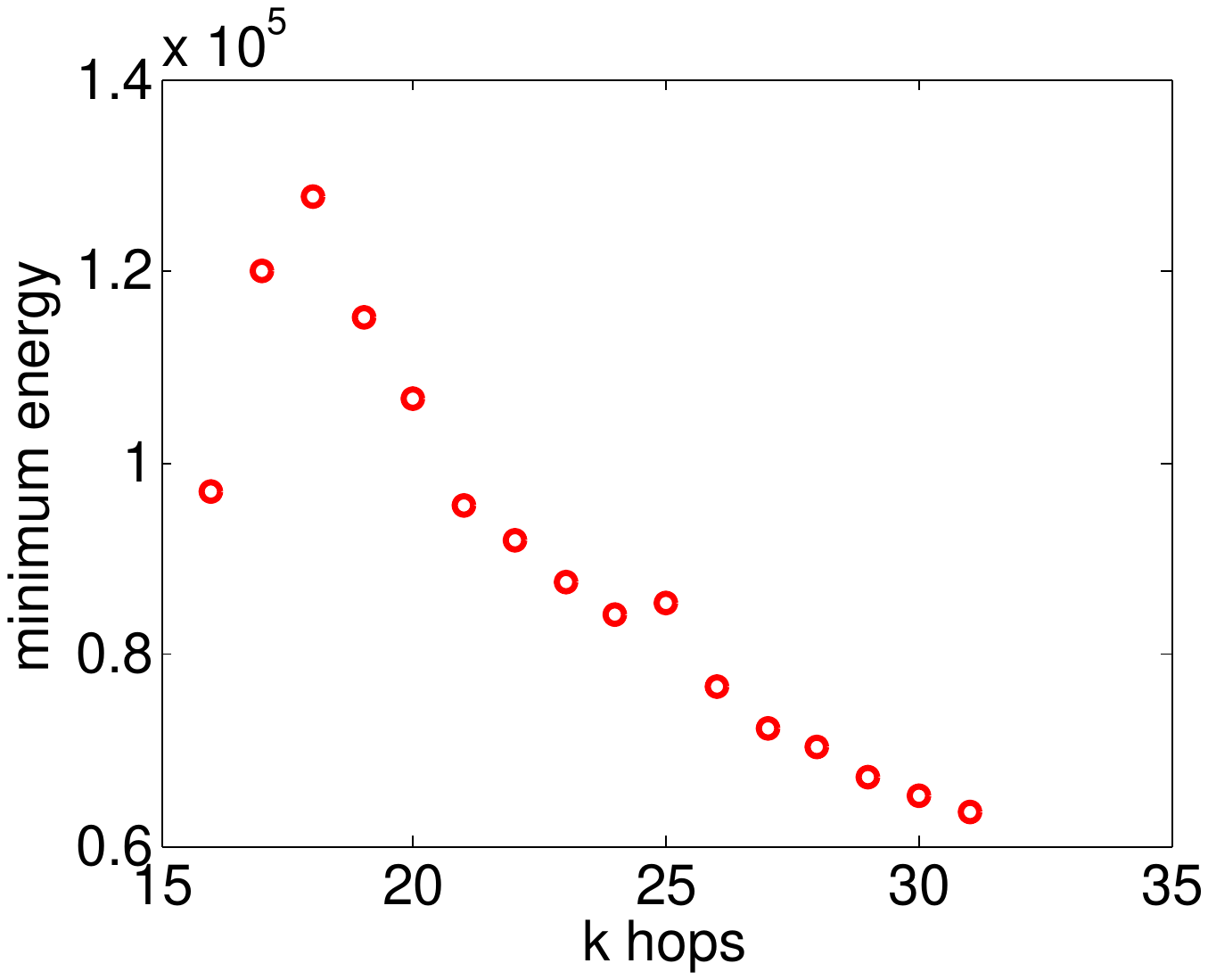}
\hfill
\includegraphics[scale=0.3, trim=5.5cm 9cm 0cm 8cm]{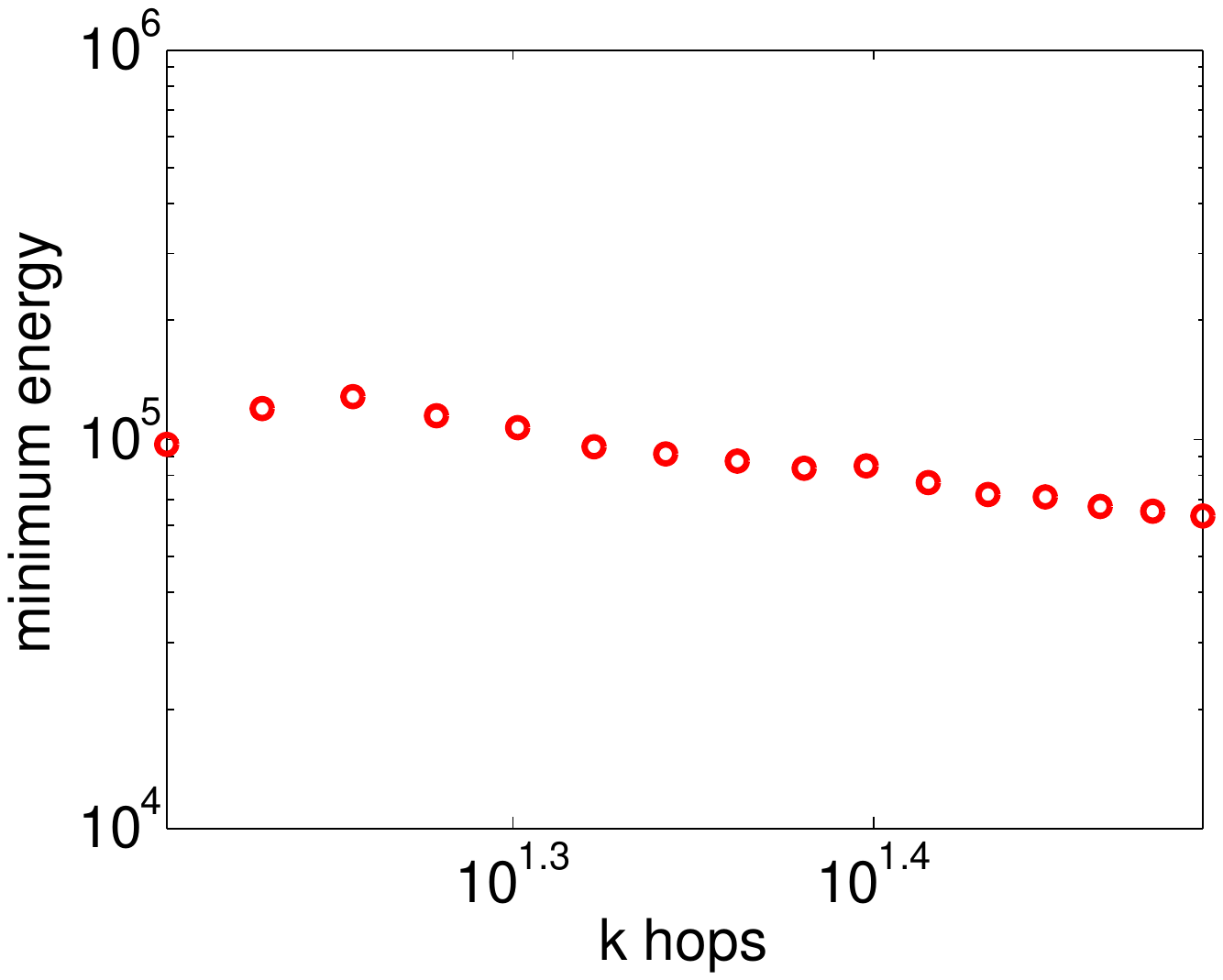}
\caption{\centering $(d_F,d_r,n)=4.3,3.3,800$}
\end{subfigure}\label{800_log}
\begin{subfigure}[t]{0.2\textwidth}  \label{1000_linear}
\includegraphics[scale=0.3, trim=3.5cm 9cm 15cm 9cm]{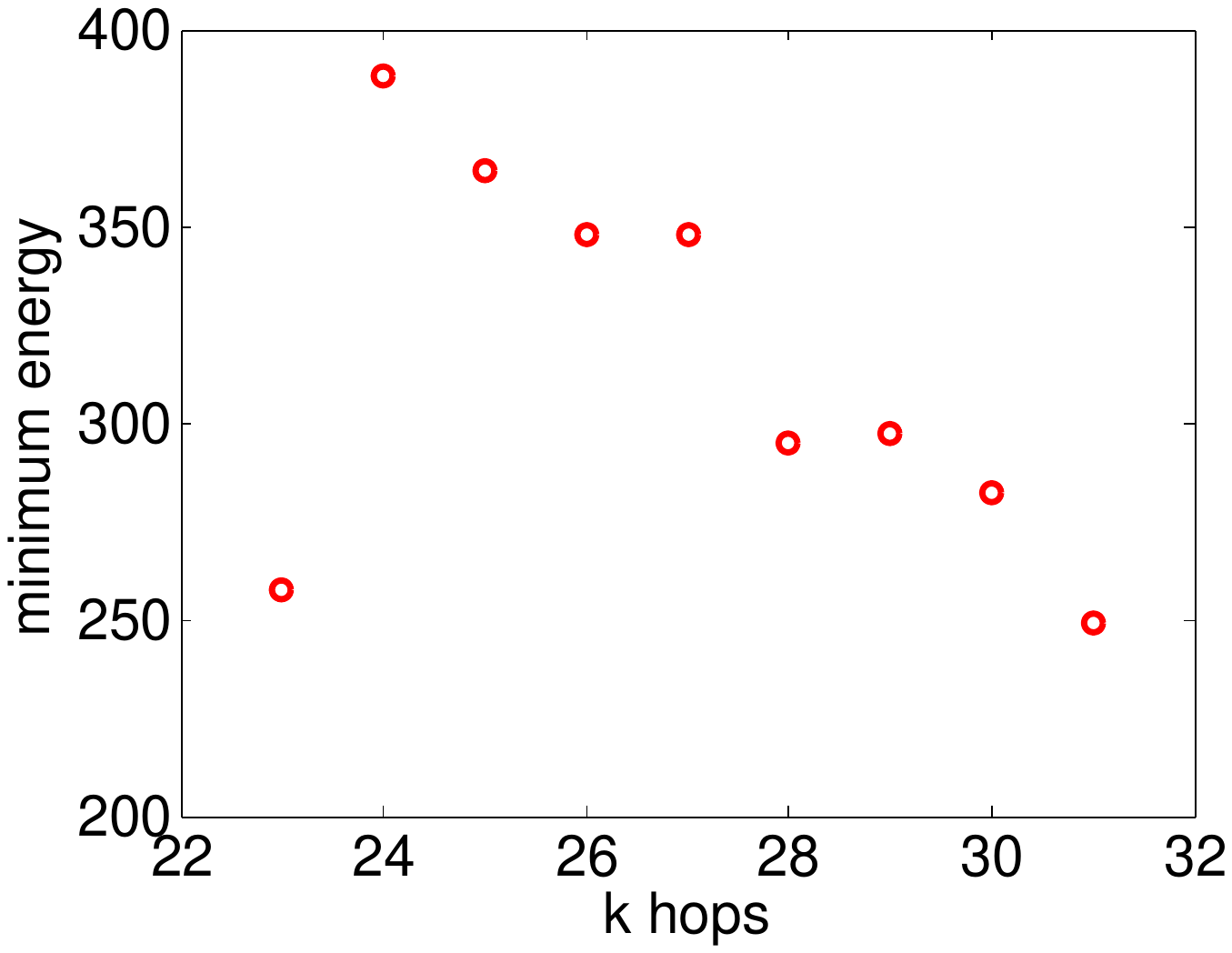}
\includegraphics[scale=0.3, trim=3.5cm 9cm 0cm 7cm]{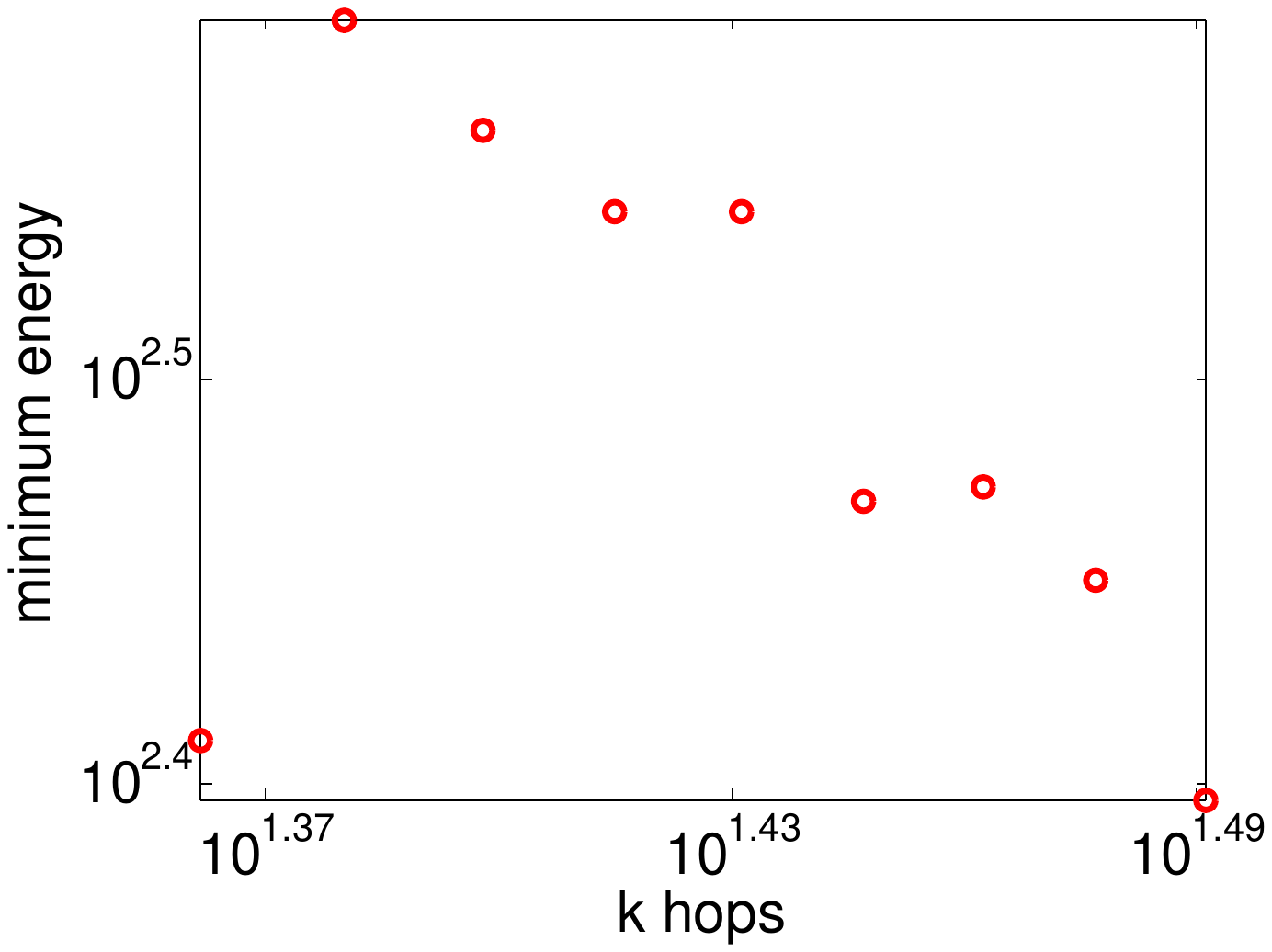}
\caption{\centering $(d_F,d_r,n)=4.3,3.3,1000$}
\end{subfigure}
\begin{subfigure}[t]{0.2\textwidth} \label{800_linear_2}
\includegraphics[scale=0.3, trim=1.5cm 9cm 0cm 9cm]{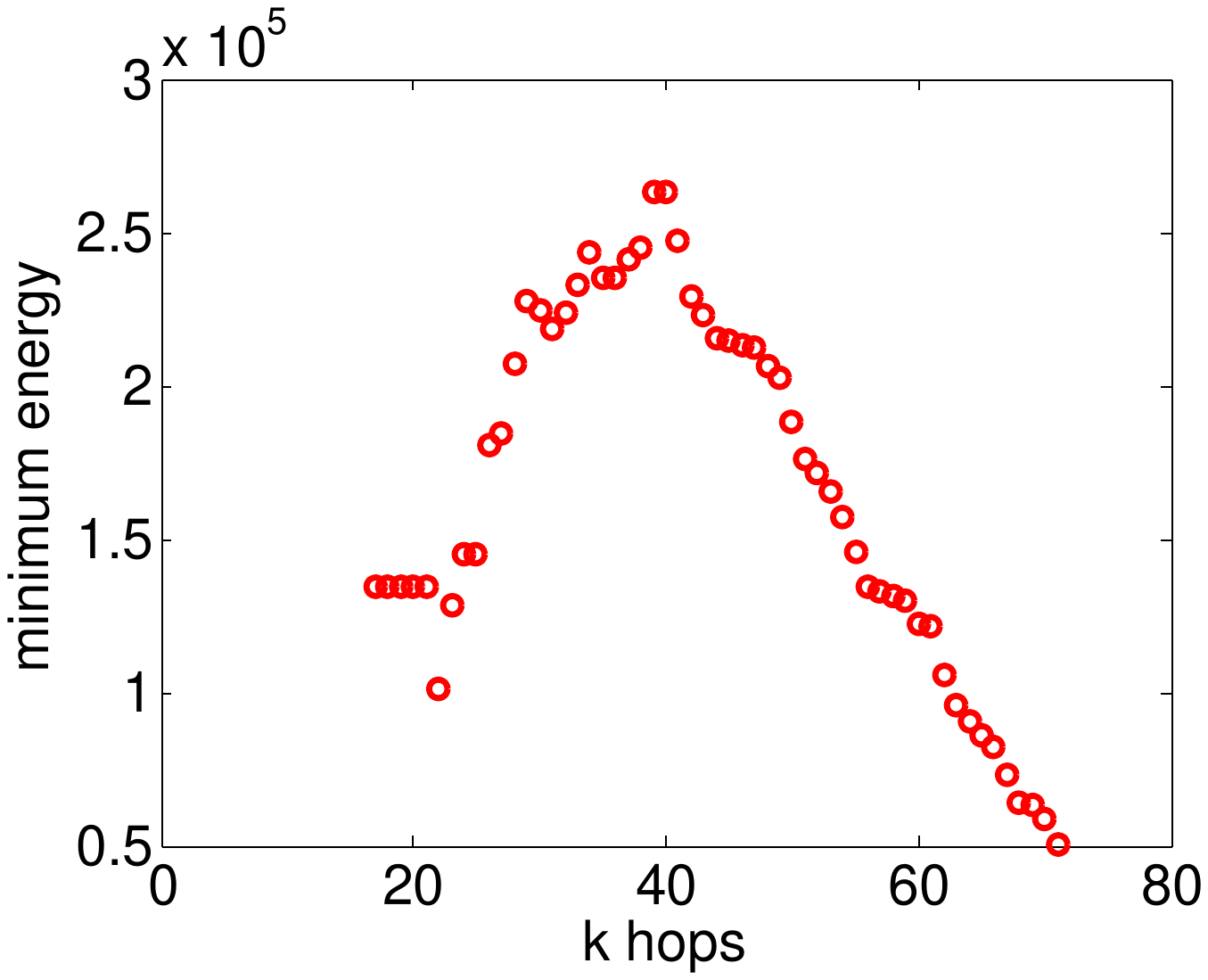}
\includegraphics[scale=0.3, trim=1.5cm 9cm 0cm 7cm]{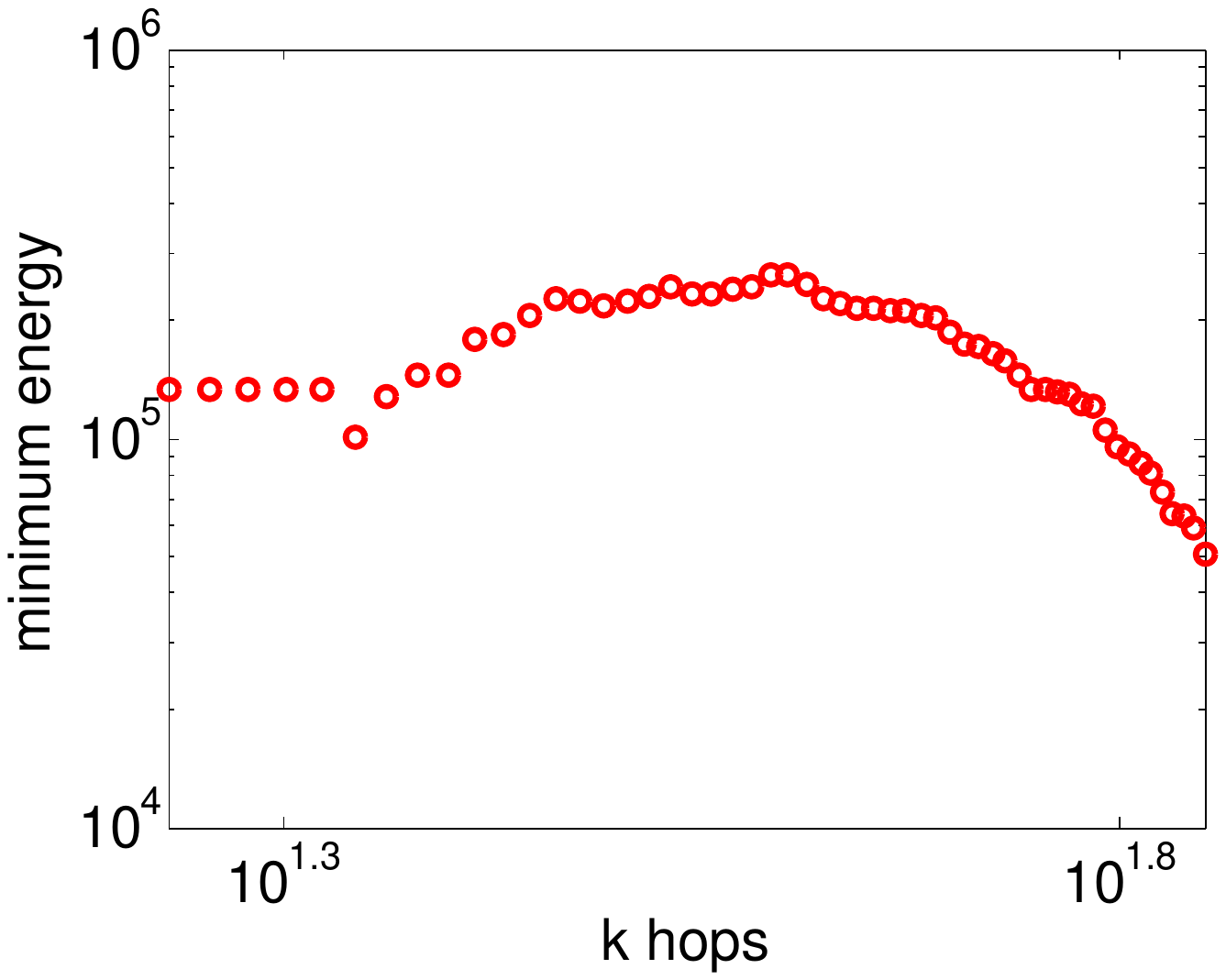}
\caption{\centering $(d_F,d_r,n)=3.3,2.3,800$}
\end{subfigure}
\begin{subfigure}[t]{0.2\textwidth} 
\includegraphics[scale=0.3, trim=0cm 19cm 0cm 3cm]{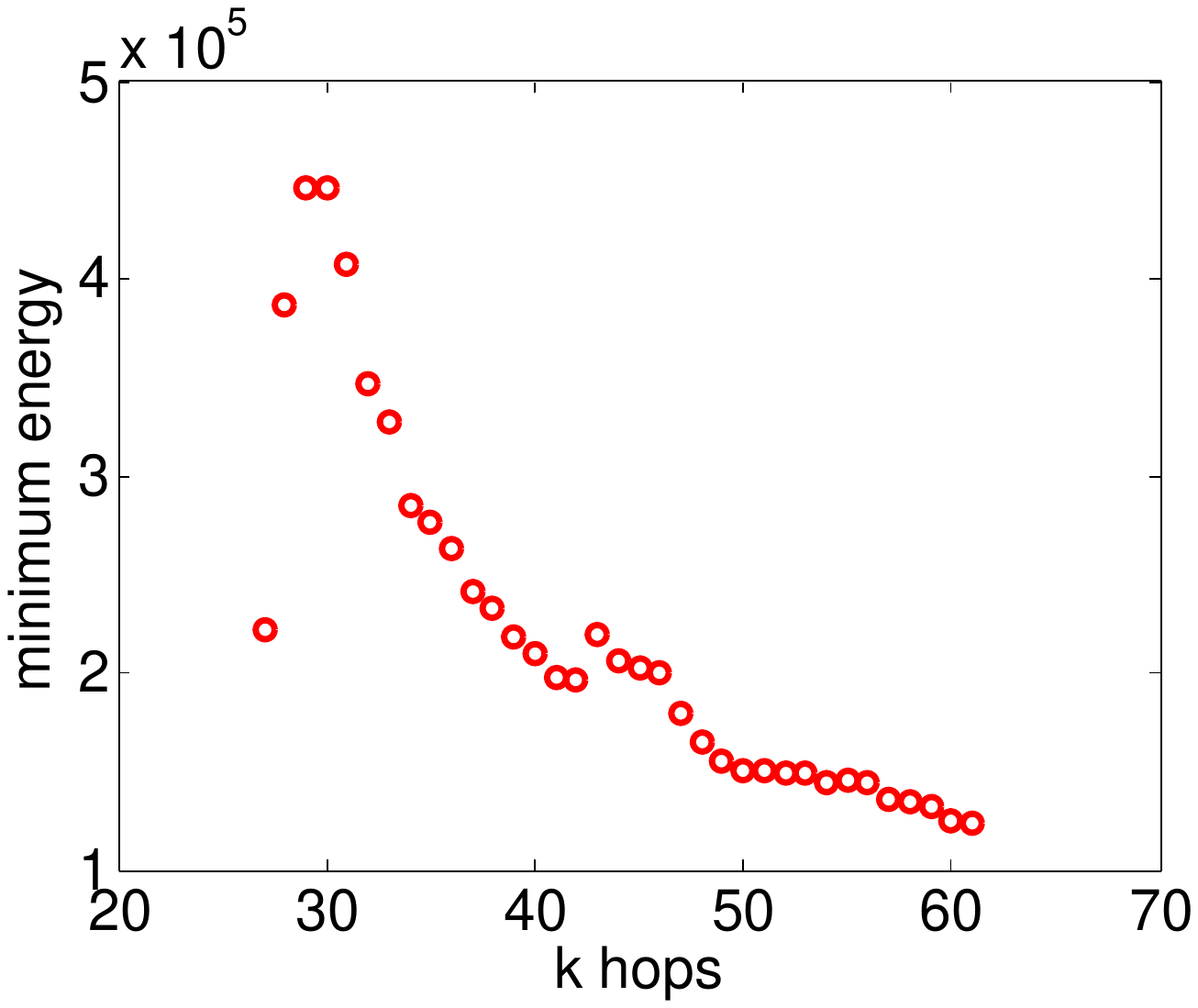}
\vspace*{0.52cm}
\includegraphics[scale=0.3, trim=-0.5cm 11cm 7cm -3cm]{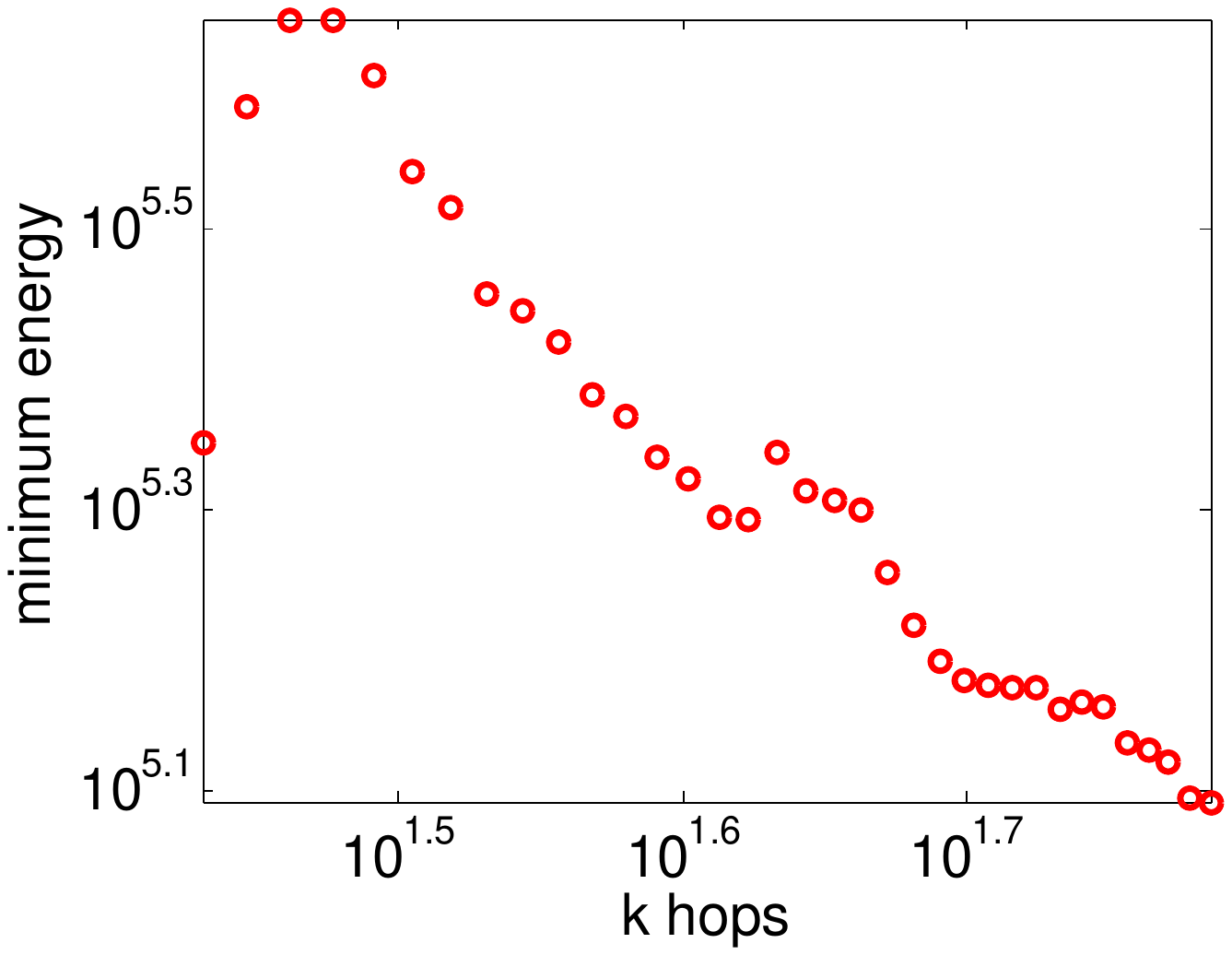}
\caption{\centering $(d_F,d_r,n)=3.3,2.3,1000$}
\end{subfigure}
\caption{Minimum accumulated end-to-end energy versus hops, averaging over 100 transmitter-receiver pairs, $\delta=2$, linear scale left side of sub-figures, logarithmic scale right side of sub-figures}
\label{fig:all_energy1}
\end{figure*}
The  results exhibited in Figure \ref{fig:all_energy1}  are obtained by computing, for each of the transmitter-receiver pair, the minimum accumulated end-to-end energy for a path smaller than $k$, then averaging over the 100 results.
The left-hand sides of the Figures \ref{fig:all_energy1}~(a) and \ref{fig:all_energy1}~(b) show the variation of the minimum path accumulated energy for the path with the increase of the number of hops in a hyperfractal setup of $d_F=4.33$ and $d_r=3$ for $n=800$ in \ref{fig:all_energy1}~(a) and $n=1000$ in \ref{fig:all_energy1}~(b). 
The figures illustrate that, indeed, allowing the hop count to grow decreases the energy considerably. The decay of the maximum accumulated energy with the allowed number of hops is even more visible in logarithmic scale in the right side of the same figures.

The decays remain substantial when changing the hyperfractal setup to $d_F=3.2$, $d_r=2.3$. Figures~\ref{fig:all_energy1}~(c) and \ref{fig:all_energy1}~(d)  show the results for $n=800$ and $n=1000$ in the new setup. The decay is dramatic when looking in logarithmic scale. Even though there can be oscillations around the linearly decreasing characteristic, as seen in Figure \ref{fig:all_energy1}~(d), left-hand side, the global behavior stays the same, decreasing, as better noticed in logarithmic scale in Figure~\ref{fig:all_energy1}~(d), right-hand side.

When changing the pathloss coefficient to $\delta=3$, the effect of Theorem~\ref{theo_energy1} remains, as illustrated in Figure~\ref{fig:all_energy2} for a hyperfractal setup of $d_F=4.33$, $d_r=3$, $n=800$ nodes.  

\begin{figure}[tb]\centering
\begin{subfigure}[t]{0.4\textwidth} 
\includegraphics[scale=0.38, trim=1cm 9cm 17cm 9cm]{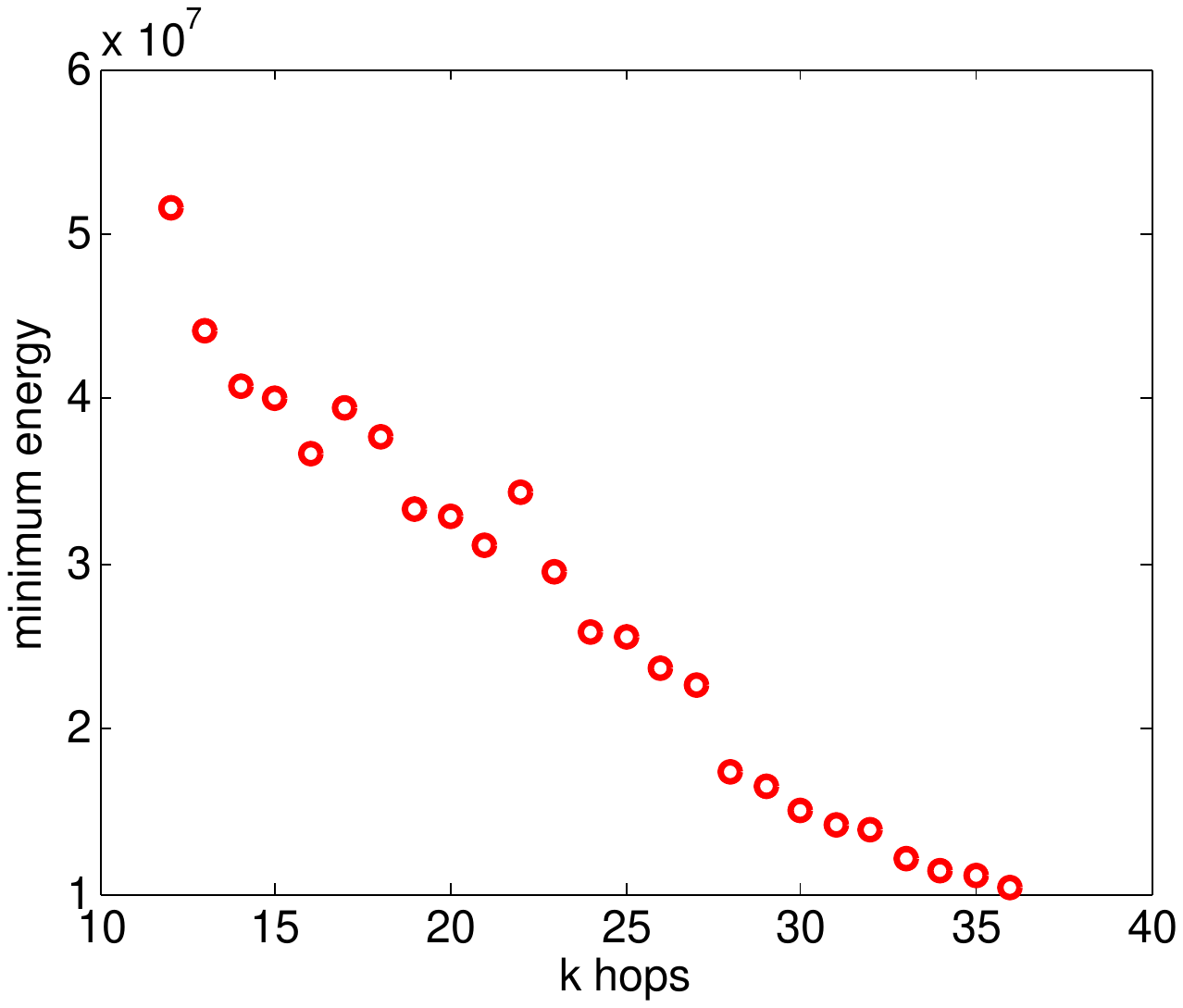}
\caption{\centering Linear}
\end{subfigure}
\begin{subfigure}[t]{0.4\textwidth} 
\includegraphics[scale=0.38, trim=1.3cm 9cm 19cm 9cm]{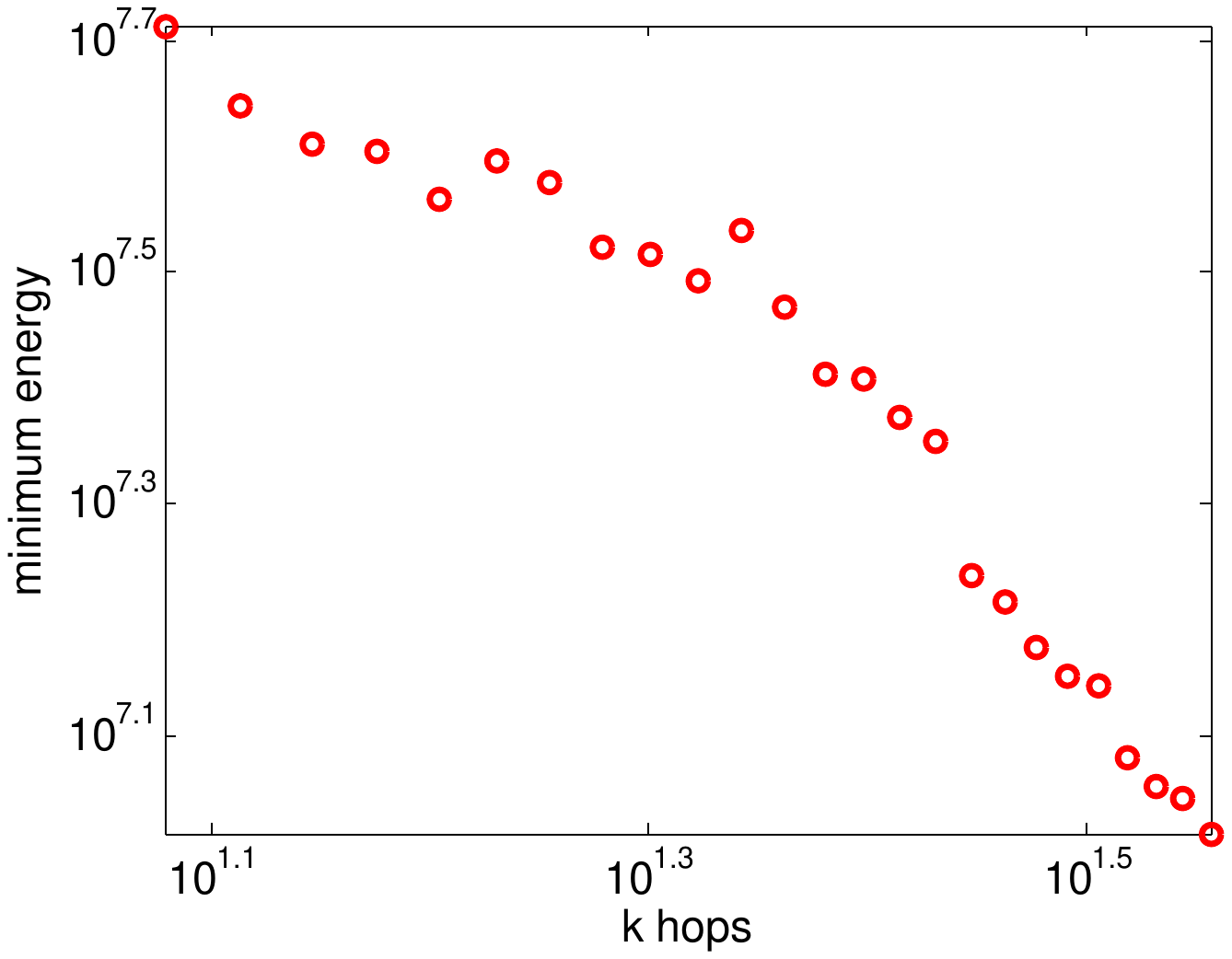}
\caption{\centering Logarithmic scale}
\end{subfigure}
\caption{Minimum accumulated end-to-end energy versus hops, averaging over 100 transmitter-receiver pairs, $\delta=3$}
\label{fig:all_energy2}

\end{figure}

To validate the results of Theorem~\ref{theo_energy2} on the variation of 
path length with the imposed constraint on  maximum energy per node, we choose randomly 100 transmitter-receiver pairs belonging to the central cross and compute the shortest path by applying a constraint on the maximum transmission energy of nodes belonging to the path. The hyperfractal setups are: nodes fractal dimension $d_F=3.3$, relays fractal dimension $d_r=2.3$, pathloss coefficient $\delta=3$ and we vary the number of nodes, $n$ to be either $n=500$ or $n=800$.  
For both values of $n$, Figure~\ref{fig:maximum}~(a) confirms that decreasing the constraint of path maximum energy increases the path length. 

\begin{figure}[tb]\centering
\begin{subfigure}[tb]{0.4\textwidth}
\includegraphics[scale=0.38,  trim={1cm 8.3cm 0cm 8.7cm}] {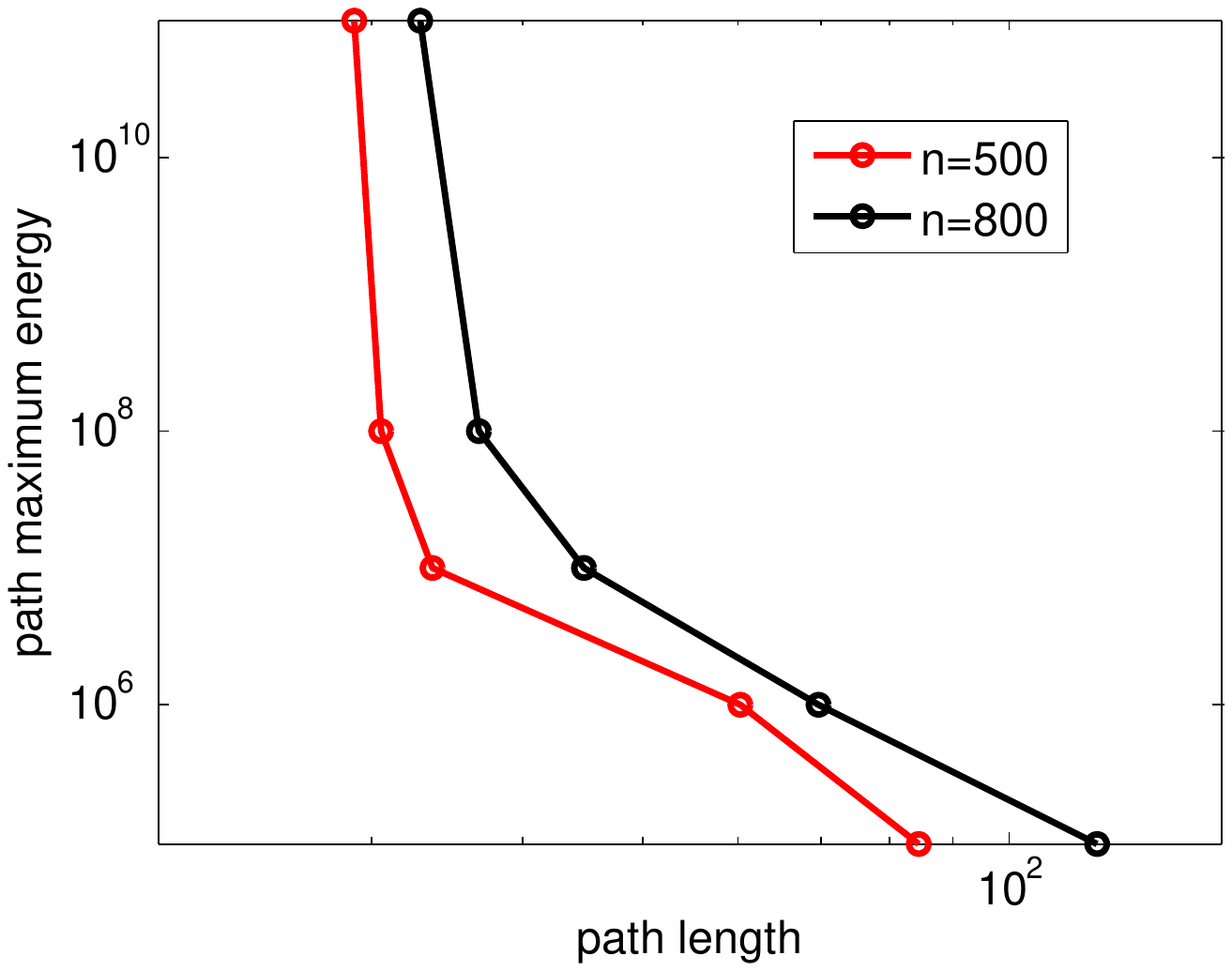}
\caption{\centering $d_F=3.3$, $d_r=2.3$} 
\vspace{-0.2cm}
\end{subfigure}
\begin{subfigure}[tb]{0.4\textwidth}
\includegraphics[scale=0.38,  trim={1cm 8.3cm 0cm 8.7cm}] {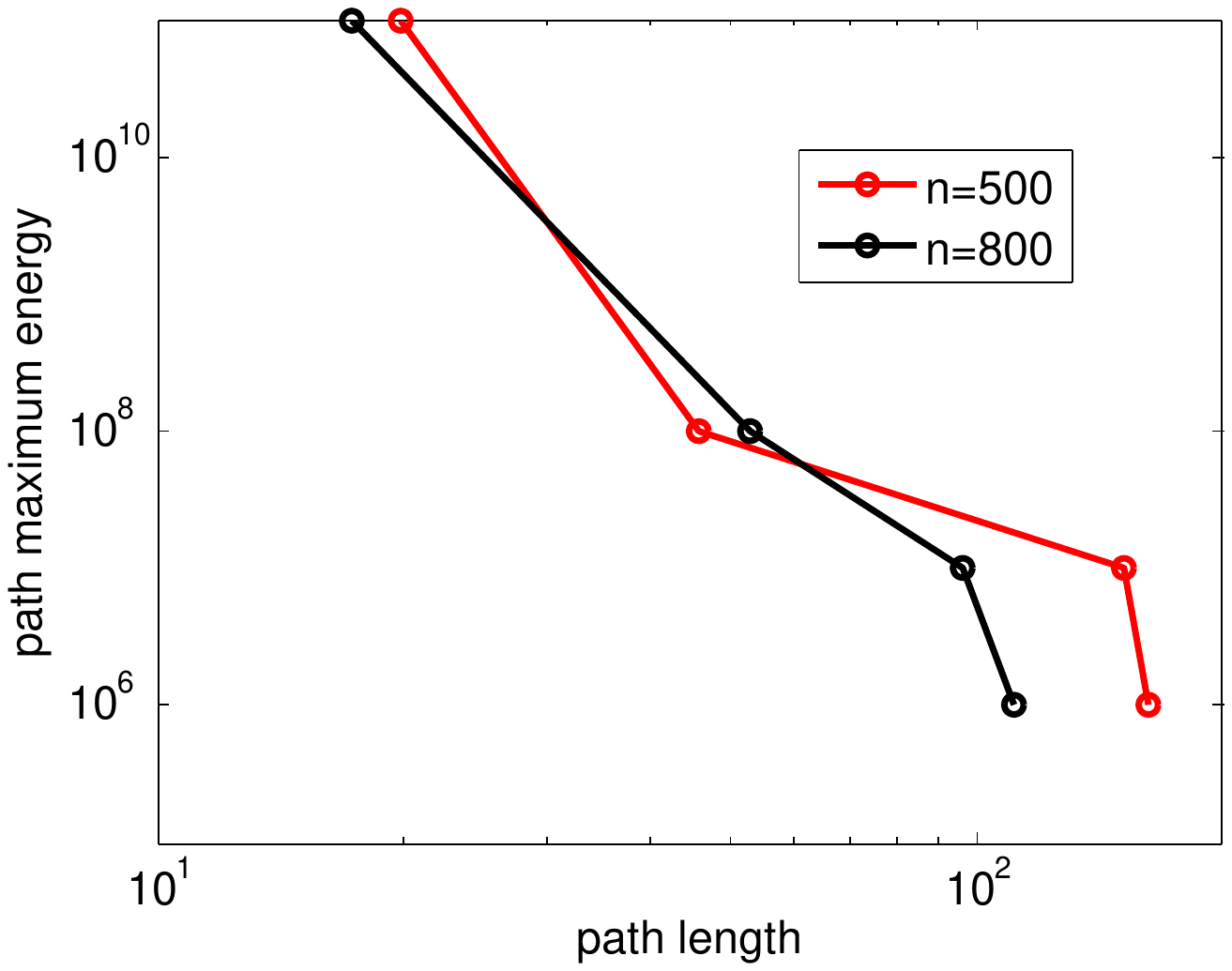}
\caption{\centering $d_F=4.33$, $d_r=3$} 
\end{subfigure}
\caption{Path Maximum Energy trade-off with delay (i.e. path length)}
\label{fig:maximum}
\vspace{-0.2cm}\end{figure}

Changing the fractal dimensions does not change the behavior, as illustrated in~Figure~\ref{fig:maximum}~(b). The hyperfractal configurations are: nodes fractal dimension $d_F=4.33$, relays fractal dimension $d_r=3$, pathloss coefficient $\delta=4$ and we vary the number of nodes, $n$ to be either $n=500$ or $n=800$. 
 Again,  making a tougher constraint on the path maximum energy leads to the increase of the path length, showing that achievable trade-offs in hyperfractal maps of nodes with RSU.

\section{Conclusion}
This paper presented results on the trade-offs between the end-to-end communication delay and energy spent on completing a transmission in millimeter-wave vehicular communications in urban settings by exploiting the ``hyperfractal'' model.  This model captures self-similarity as an environment characteristic. The self-similar characteristic of the road-side infrastructure has also been incorporated. 

Analytical bounds have been derived 
for the end-to-end communication hop count under the constraints of total accumulated energy, and maximum energy per node, exhibiting the achievable trade-offs in a hyperfractal network. The work presented a lower bound on the network throughput capacity with constraints on path energy. 
Further examples of model fitting with data have been given. The analytical results have been validated using a discrete-time event-based simulator developed in Matlab.


%

\appendices
\section{Proofs}
\subsection{Proof of Lemma \ref{lemma:existence}}
\begin{proof}
Let $N_H(n)$ be the number of nodes contained in the street of level $H$.

Let $z$ be a real number. By Chebyshev's inequality, we have:
\begin{equation*}\nonumber
\mathbb{E}[e^{zN_H(n)}]=\left(1+(e^z-1)\lambda_H\right)^n
\end{equation*}
If $z>0$:
$$
P\left(N_H(n)<\frac{n\lambda_H}{2}\right) = P\left(e^{-zN_H(n)}>e^{zn\lambda_H/2}\right)
\le \frac{\mathbb{E}[e^{-zN_H(n)}]}{e^{-zn\lambda_H/2}}
$$
Therefore
\begin{equation*}\nonumber
\frac{\mathbb{E}[e^{-zN_H(n)}]}{e^{-zn\lambda_H}/2}=\exp\left(n\left(\log\left(1+(e^{-z}-1)\lambda_H\right)+z\lambda_H/2\right)\right).
\end{equation*}

For $|z|$ bounded there exists $b>0$ such that $|e^z-1|\le b|z|$ and there exists $c$ such that $e^z-1\le z+cz^2$. For $|x|$ bounded there exists $d$ such that $\log(1+x)\le x-cx^2$. From these steps we obtain that, for sufficiently small $|z|$, one has:
\begin{eqnarray*}
\log\left(1+(e^{-z}-1)\lambda_H\right)+&z\frac{\lambda_H}{2}\le-z\frac{\lambda_H}{2}
+b\lambda_H z^2-c\lambda^2_H z^2
\le -a\lambda_H.
\end{eqnarray*}
which settles that 
\begin{equation}
\frac{\mathbb{E}[e^{-zN_H(n)}]}{e^{-zn\lambda_H/2}}\le e^{-an\lambda_H}.
\end{equation}
The proof of the second part of the lemma proceeds via similar reasoning, by using the inequality:
\begin{equation}
P\left(N_H(n)>2n\lambda_H\right)\le \frac{\mathbb{E}[e^{zN_H(n)}]}{e^{2zn\lambda_H}}.
\end{equation}
\end{proof}

\ifCLASSOPTIONcaptionsoff
  \newpage
\fi



%

%

\bibliographystyle{IEEEtran}
\bibliography{mybib}

\end{document}